\documentclass[aps,10pt,prl,groupedaddress,superscriptaddress,twocolumn,notitlepage,nofootinbib]{revtex4-2}
\pdfoutput=1

\usepackage{newtxtext,newtxmath}
\makeatletter
\let\rel@kern\relax
\let\widebar\relax
\makeatother

\usepackage[latin1]{inputenc}
\usepackage{amsthm}
\usepackage{amssymb}
\usepackage{amsmath}
\usepackage{bbold}
\usepackage{bbm}
\usepackage[pdftex, 
colorlinks=true, allcolors=MidnightBlue!70!black!85!TealBlue]{hyperref}
\usepackage{braket}
\usepackage{dsfont}
\usepackage{mathdots}
\usepackage{mathtools}
\usepackage{enumerate}
\usepackage[shortlabels]{enumitem}
\usepackage{csquotes}
\usepackage[cal=boondox]{mathalfa}
\usepackage{graphicx}
\usepackage{stackengine}
\usepackage{scalerel}
\usepackage{tensor}       
\usepackage[dvipsnames]{xcolor}
\usepackage{array}
\usepackage{makecell}
\newcolumntype{x}[1]{>{\centering\arraybackslash}p{#1}}
\usepackage{tikz}
\usetikzlibrary{shapes.geometric, shapes.misc, positioning, arrows, arrows.meta, decorations.pathreplacing, decorations.pathmorphing, patterns, angles, quotes, calc}
\usepackage{booktabs}
\usepackage{xfrac}
\usepackage{siunitx}
\usepackage{centernot}
\usepackage{comment}
\usepackage{chngcntr}
\usepackage{titlesec}

\titleformat{\section}{\centering\bfseries\boldmath\MakeUppercase}{\thesection.}{0.5em}{}
\titleformat{\subsection}{\centering\bfseries\boldmath}{\thesubsection.}{0.5em}{}

\newtheorem{thm}{Theorem}
\newtheorem*{thm*}{Theorem}
\newtheorem{prop}[thm]{Proposition}
\newtheorem*{prop*}{Proposition}
\newtheorem{lemma}[thm]{Lemma}
\newtheorem*{lemma*}{Lemma}
\newtheorem{cor}[thm]{Corollary}
\newtheorem*{cor*}{Corollary}
\newtheorem{cj}[thm]{Conjecture}
\newtheorem*{cj*}{Conjecture}
\newtheorem{Def}[thm]{Definition}
\newtheorem*{Def*}{Definition}

\newtheorem*{question*}{Question}

\newtheorem*{problem*}{Problem}

\makeatletter
\def\thmhead@plain#1#2#3{%
  \thmname{#1}\thmnumber{\@ifnotempty{#1}{ }\@upn{#2}}%
  \thmnote{ {\the\thm@notefont#3}}}
\let\thmhead\thmhead@plain
\makeatother

\theoremstyle{definition}
\newtheorem{rem}[thm]{Remark}
\newtheorem*{note}{Note}

\newtheorem{manualthminner}{Theorem}
\newenvironment{manualthm}[1]{%
  \renewcommand\themanualthminner{#1}%
  \manualthminner
}{\endmanualthminner}

\newtheorem{manualpropinner}{Proposition}
\newenvironment{manualprop}[1]{%
  \renewcommand\themanualpropinner{#1}%
  \manualpropinner \it
}{\endmanualpropinner}

\newtheorem{manualcorinner}{Corollary}
\newenvironment{manualcor}[1]{%
  \renewcommand\themanualcorinner{#1}%
  \manualcorinner \it
}{\endmanualcorinner}

\newcommand{\bb}{\begin{equation}\begin{aligned}\hspace{0pt}}
\newcommand{\bbb}{\begin{equation*}\begin{aligned}}
\newcommand{\ee}{\end{aligned}\end{equation}}
\newcommand{\eee}{\end{aligned}\end{equation*}}
\newcommand\floor[1]{\lfloor#1\rfloor}
\newcommand\ceil[1]{\left\lceil#1\right\rceil}
\newcommand{\eqt}[1]{\stackrel{\mathclap{\mbox{\scriptsize #1}}}{=}}
\newcommand{\leqt}[1]{\stackrel{\mathclap{\mbox{\scriptsize #1}}}{\leq}}

\newcommand{\geqt}[1]{\stackrel{\mathclap{\mbox{\scriptsize #1}}}{\geq}}
\newcommand{\ketbra}[1]{\ket{#1}\!\!\bra{#1}}

\newcommand{\ketbraa}[2]{\ket{#1}\!\!\bra{#2}}

\newcommand{\sumno}{\sum\nolimits}

\newcommand{\e}{\varepsilon}
\renewcommand{\epsilon}{\varepsilon}

\newcommand{\tcr}[1]{{\color{red!85!black} #1}}
\newcommand{\tcb}[1]{{\color{blue!80!black} #1}}

\newcommand{\id}{\mathds{1}}
\newcommand{\R}{\mathds{R}}
\newcommand{\N}{\mathds{N}}

\newcommand{\C}{\mathds{C}}

\newcommand{\cptp}{\mathrm{CPTP}}

\newcommand{\locc}{\mathrm{LOCC}}

\newcommand{\ppt}{\mathrm{PPT}}

\newcommand{\PPT}{\pazocal{P\!P\!T}}

\DeclareMathOperator{\Tr}{Tr}

\DeclareMathAlphabet{\pazocal}{OMS}{zplm}{m}{n}

\newcommand{\HH}{\pazocal{H}}

\newcommand{\NN}{\pazocal{N}}

\newcommand{\XX}{\pazocal{X}}

\newcommand{\FF}{\pazocal{F}}

\newcommand{\lsmatrix}{\left(\begin{smallmatrix}}
\newcommand{\rsmatrix}{\end{smallmatrix}\right)}

\newcommand{\deff}[1]{\textbf{\emph{#1}}}

\stackMath

\stackMath

\makeatletter
\newcommand*\rel@kern[1]{\kern#1\dimexpr\macc@kerna}
\newcommand*\widebar[1]{%
  \begingroup
  \def\mathaccent##1##2{%
    \rel@kern{0.8}%
    \overline{\rel@kern{-0.8}\macc@nucleus\rel@kern{0.2}}%
    \rel@kern{-0.2}%
  }%
  \macc@depth\@ne
  \let\math@bgroup\@empty \let\math@egroup\macc@set@skewchar
  \mathsurround\z@ \frozen@everymath{\mathgroup\macc@group\relax}%
  \macc@set@skewchar\relax
  \let\mathaccentV\macc@nested@a
  \macc@nested@a\relax111{#1}%
  \endgroup
}

\newcommand{\fakepart}[1]{
 \par\refstepcounter{part}
  \sectionmark{#1}
}

\counterwithin*{equation}{part}
\counterwithin*{thm}{part}
\counterwithin*{figure}{part}

\tikzset{meter/.append style={draw, inner sep=10, rectangle, font=\vphantom{A}, minimum width=30, line width=.8, path picture={\draw[black] ([shift={(.1,.3)}]path picture bounding box.south west) to[bend left=50] ([shift={(-.1,.3)}]path picture bounding box.south east);\draw[black,-latex] ([shift={(0,.1)}]path picture bounding box.south) -- ([shift={(.3,-.1)}]path picture bounding box.north);}}}
\tikzset{roundnode/.append style={circle, draw=black, fill=gray!20, thick, minimum size=10mm}}
\tikzset{squarenode/.style={rectangle, draw=black, fill=none, thick, minimum size=10mm}}

\definecolor{Blues5seq1}{RGB}{239,243,255}
\definecolor{Blues5seq2}{RGB}{189,215,231}
\definecolor{Blues5seq3}{RGB}{107,174,214}
\definecolor{Blues5seq4}{RGB}{49,130,189}
\definecolor{Blues5seq5}{RGB}{8,81,156}

\definecolor{Greens5seq1}{RGB}{237,248,233}
\definecolor{Greens5seq2}{RGB}{186,228,179}
\definecolor{Greens5seq3}{RGB}{116,196,118}
\definecolor{Greens5seq4}{RGB}{49,163,84}
\definecolor{Greens5seq5}{RGB}{0,109,44}

\definecolor{Reds5seq1}{RGB}{254,229,217}
\definecolor{Reds5seq2}{RGB}{252,174,145}
\definecolor{Reds5seq3}{RGB}{251,106,74}
\definecolor{Reds5seq4}{RGB}{222,45,38}
\definecolor{Reds5seq5}{RGB}{165,15,21}

\allowdisplaybreaks

\makeatletter 
\renewcommand\onecolumngrid{
\do@columngrid{one}{\@ne}%
\def\set@footnotewidth{\onecolumngrid}
\def\footnoterule{\kern-6pt\hrule width 1.5in\kern6pt}%
}
\makeatother 
\usepackage{anyfontsize}

\makeatletter
\newcommand{\raisemath}[1]{\mathpalette{\raisem@th{#1}}}
\newcommand{\raisem@th}[3]{\raisebox{#1}{$#2#3$}}
\makeatother

\newcommand{\exact}{\mathrm{exact}}
\newcommand{\ecost}{E^{\raisemath{0.5pt}{\exact}}_{ c,\,\ppt}}
\newcommand{\LL}{\pazocal{L}}

\usepackage{pgfplots}
\pgfplotsset{width=10cm,compat=1.9}
\usetikzlibrary{decorations.markings}

\usepackage{algorithm}
\usepackage{algpseudocode}
\algrenewcommand\algorithmicrequire{\textbf{Input:}}
\algrenewcommand\algorithmicensure{\textbf{Output:}}

\let\tcb\relax
\let\tcr\relax

\begin{document}

\title{Computable entanglement cost under positive partial transpose operations}

\author{Ludovico Lami}
\email{ludovico.lami@gmail.com}
\affiliation{Scuola Normale Superiore, Piazza dei Cavalieri 7, 56126 Pisa, Italy}
\affiliation{QuSoft, Science Park 123, 1098 XG Amsterdam, the Netherlands}
\affiliation{Korteweg--de Vries Institute for Mathematics, University of Amsterdam, Science Park 105-107, 1098 XG Amsterdam, the Netherlands}
\affiliation{Institute for Theoretical Physics, University of Amsterdam, Science Park 904, 1098 XH Amsterdam, the Netherlands}

\author{Francesco Anna Mele}
\email{francesco.mele@sns.it}
\affiliation{NEST, Scuola Normale Superiore and Istituto Nanoscienze, Consiglio Nazionale delle Ricerche, Piazza dei Cavalieri 7, IT-56126 Pisa, Italy}

\author{Bartosz Regula}
\email{bartosz.regula@gmail.com}
\affiliation{Mathematical Quantum Information RIKEN Hakubi Research Team, RIKEN Cluster for Pioneering Research (CPR) and RIKEN Center for Quantum Computing (RQC), Wako, Saitama 351-0198, Japan}

\begin{abstract}
Quantum information theory is plagued by the problem of regularisations, which require the evaluation of formidable asymptotic quantities. This makes it computationally intractable to gain a precise quantitative understanding of the ultimate efficiency of key operational tasks such as entanglement manipulation. Here we consider the problem of computing the asymptotic entanglement cost of preparing noisy quantum states under quantum operations with positive partial transpose (PPT). 
\tcr{By means of an analytical example, a previously claimed solution to this problem is shown to be incorrect. Building on a previous characterisation of the PPT entanglement cost in terms of a regularised formula, we construct instead a hierarchy of semi-definite programs that bypasses the issue of regularisation altogether, and converges to the true asymptotic value of the entanglement cost.}
Our main result establishes that this convergence happens exponentially fast, thus yielding an efficient algorithm that approximates the cost up to an additive error $\e$ in time $\operatorname{poly}\!\big(D,\,\log(1/\e)\big)$, where $D$ is the underlying Hilbert space dimension. To our knowledge, this is the first time that an asymptotic entanglement measure is shown to be efficiently computable despite no closed-form formula being available. 
\end{abstract}

\maketitle

\let\oldaddcontentsline\addcontentsline
\renewcommand{\addcontentsline}[3]{}
\fakepart{Main text}

\noindent \textbf{\em Introduction.}--- Quantum Shannon theory studies the fundamental limitations on the manipulation of quantum information in the presence of external noise. Calculating those limits often involves computing certain functions that encapsulate the ultimate capabilities of information carriers. Paradigmatic examples include the various capacities of quantum channels, such as the classical~\cite{Holevo-S-W, H-Schumacher-Westmoreland}, quantum~\cite{Lloyd-S-D, L-Shor-D, L-S-Devetak}, private~\cite{L-S-Devetak, CWY}, and entanglement-assisted~\cite{entanglement-assisted, Bennett2002} capacities, but also the operational entanglement measures that tell us how much entanglement can be extracted from a given bipartite quantum state, i.e.\ the distillable entanglement~\cite{Bennett-distillation, Bennett-distillation-mixed, Bennett-error-correction, devetak2005, DNE-distillable}, and vice versa, how much entanglement must be invested to \emph{create} that state~\cite{Hayden-EC, faithful-EC, Buscemi2011, EC-infinite}. This latter quantity, the entanglement cost, is the main focus of this work.

With the sole exception of the entanglement-assisted capacity, 
all of the above functions are expressed by \emph{regularised formulas}, i.e.\ formulas that involve an explicit limit $n\to\infty$ over the number of uses of the channel or \tcb{the} available copies of the state. For example, by the Lloyd--Shor--Devetak theorem~\cite{Lloyd-S-D, L-Shor-D, L-S-Devetak} the quantum capacity of a channel $\NN$ equals $Q(\NN) = \lim_{n\to\infty} \frac1n\, I_c\big(\NN^{\otimes n}\big)$, where $I_c(\NN)$ is the `coherent information' of $\NN$, and $\NN^{\otimes n}$ represents $n$ parallel uses of $\NN$. In stark contrast with classical information theory, for quantum channels it holds in general that $I_c(\NN^{\otimes n}) \neq n I_c(\NN)$, meaning that evaluating the limit cannot be avoided. Such non-additivity is a fundamental feature of most settings encountered in quantum information~\cite{Werner-symmetry, Shor2004, Hayden-p>1, Cubitt-p-->0, Hastings2008, superactivation}. Analogously, the entanglement cost of preparing a state $\rho=\rho_{AB}$ using local operations and classical communication (LOCC) is given~by 
\bb
E_{c,\,\locc}(\rho) = \lim_{n\to\infty} \frac1n\, E_f\big(\rho^{\otimes n}\big)\, ,
\label{cost_LOCC}
\ee
where $E_f$ is 
the `entanglement of formation'~\cite{Bennett-error-correction, Wootters1998}. The precise nature of these formulas is not so important here; what is important, however, is that the regularisation $n\to\infty$ makes them analytically hard to control and computationally intractable. Indeed, on the one hand the dimension of the quantum system on which $\NN^{\otimes n}$ acts, or to which $\rho^{\otimes n}$ pertains, is exponential in $n$, quickly rendering numerical calculations infeasible as $n$ grows; on the other, there is no guarantee on the quality of the approximation obtained by stopping at the $n^{\text{th}}$ level 
in any of these formulas --- for instance, an unbounded $n$ may be required to even check that the quantum capacity is non-zero~\cite{Cubitt2015}. The regularisation thus appears to be an omnipresent curse that stifles almost every attempt to quantitatively understand the ultimate limitations of quantum information manipulation. 


But is that really so? In this work we show how to overcome this fundamental obstacle in a specific case, by efficiently calculating a type of entanglement cost --- expressed as a regularised quantity --- on all quantum states. To this end, we look at a problem that has been studied by many authors~\cite{Martin-exact-PPT, Xin-exact-PPT, Wang2023, Gour2020, Gour2019}, but for which a full solution had not been found prior to our work: namely, zero-error asymptotic entanglement 
cost under 
`positive partial transpose' (PPT) operations. 
The roots of PPT transformations lie in the fundamental connection between entanglement and partial transposition identified by Peres~\cite{PeresPPT} and the Horodeckis~\cite{HorodeckiPPT, Horodecki-PPT-entangled}. In the context of entanglement manipulation, PPT operations were introduced in the pioneering works of Rains~\cite{Rains1999, Rains2001} as a mathematically natural relaxation of the much more complicated set of LOCC operations, providing a convenient way to gain some insights into the latter and, importantly, giving hope for an easier 
understanding of some of the fundamental limits of entanglement transformations. Because of this, they attracted significant attention in the operational study of quantum entanglement, but even in this technically simpler setting the fundamental questions in entanglement manipulation remained unsolved.

By introducing an efficient algorithm for the evaluation of the asymptotic PPT entanglement cost, we solve this long-standing problem and exhibit an entanglement measure that is both \emph{computable} and \emph{operationally meaningful} for general --- pure and mixed --- quantum states. This can be contrasted with other entanglement measures, which are either defined in terms of operational tasks and thus suffer from the problem of regularisation, making efficient evaluation impossible, or they are simply abstract mathematical constructions with no precise operational meaning. Prior to our work, the PPT entanglement cost was only known to be computable for a specific class of quantum states~\cite{Martin-exact-PPT} (including pure states, Gaussian states, Werner states, and two-qubit states~\cite{Ishizaka2004}), with its value given by a celebrated entanglement measure known as the \emph{logarithmic negativity}~\cite{Martin-exact-PPT,negativity, plenioprl}. The logarithmic negativity is one of the most widely employed entanglement measures, having found numerous applications in all corners of quantum physics~\cite{Ent_neg_Calabrese, Calabrese_2013, Calabrese_2014, Ent_neg_Lee, Hoogeveen_2015a, Eisler_2014, Ent_neg_random_spin, Dong2024}, from quantum field theory~\cite{Ent_neg_Calabrese} to condensed matter~\cite{Ent_neg_Lee}. The success of the logarithmic negativity is due to two main reasons: first, its efficient computability~\cite{negativity}, which makes it highly suitable for numerical applications, and second, its \emph{partial} operational interpretation, which guarantees that it coincides with the zero-error PPT entanglement cost for this specific class of quantum states~\cite{Martin-exact-PPT}. While encouraging, this partial operational interpretation is arguably not satisfactory if this measure is to be applied widely in quantum physics. Indeed, beyond this restricted class of states the negativity is not known to be an operational quantity, and it is unclear whether e.g.\ the states considered in~\cite{Calabrese_2013, Calabrese_2014, Dong2024} genuinely belong to the class in question.


It was recently claimed that the PPT entanglement cost can be computed exactly for all states~\cite{Xin-exact-PPT, Wang2023}, but --- as we show below (see also~\cite{Xin-Mark-errata}) --- this claim relies on some erroneous assertions. In this work, building on the partial results of~\cite{Xin-exact-PPT, Wang2023}, we find the correct generalisation of the logarithmic negativity that enjoys the operational interpretation of being the zero-error PPT entanglement cost for \emph{all} quantum states. Crucially, 
our generalisation is also efficiently computable, requiring almost the same computational time as the logarithmic negativity. This opens the door to numerous applications in quantum physics, serving as a fully operationally motivated alternative to the negativity as a computable entanglement quantifier. 

\medskip \noindent 
\textbf{\em Zero-error PPT entanglement cost.\,}--- 
The goal of 
zero-error PPT entanglement dilution is to prepare $n$ copies of a given bipartite quantum state $\rho = \rho_{AB}$ by consuming as few singlets $\Phi_2$ (i.e.\ two-qubit maximally entangled states) as possible and using PPT operations only. 
Here, a channel $\Lambda$ is called PPT if its partial transposition $\Gamma \circ \Lambda \circ \Gamma$ is also a valid quantum channel, where $\Gamma$ is the partial transpose operation~\cite{PeresPPT} defined as $\Gamma(X_A\otimes Y_B) = X_A \otimes Y_B^\intercal$.
We say that a number $R$ is an \emph{achievable rate} if for all sufficiently large $n$ there exists a PPT operation $\Lambda_n$ with the property that $\Lambda_n\big(\Phi_2^{\otimes \floor{Rn}}\big) = \rho^{\otimes n}$. By definition, the zero-error PPT entanglement cost of $\rho$, denoted as $\ecost(\rho)$, is the infimum of all achievable rates $R$. In this work, we show how to compute $\ecost(\rho)$ for any quantum state $\rho$. 

In a nutshell, the key reason why many authors~\cite{Martin-exact-PPT, Xin-exact-PPT, Wang2023, Gour2020, Gour2019} have been interested in the problem of entanglement dilution under PPT operations is that it provides a more tractable model for the fundamental problem of entanglement dilution under LOCC operations, whose underlying figure of merit, given by~\eqref{cost_LOCC}, is computationally inaccessible in most cases of interest. 
In general, PPT operations are a superset of LOCC, but any PPT channel can be implemented in a stochastic manner by LOCC together with the assistance of PPT states~\cite{Rains2001,cirac_2001}; since it is known that PPT states possess only a weak form of entanglement 
\tcb{that} cannot be distilled~\cite{HorodeckiBound}, such states can be considered as a `cheap' resource, thus providing intuition for the PPT operations being a prudent relaxation of the power of LOCC. 
Such PPT-based approaches have \tcb{also} attracted significant attention 
in other parts of entanglement theory~\cite{Rains1999, Rains2001, Eggeling2001, irreversibility-PPT}.

Importantly, the optimal performance achievable under PPT operations also establishes bounds and no-go limits on what can be achieved under LOCC in practice. These bounds are often the tightest available, as evidenced in the contexts of channel capacities~\cite{Wang_2018, Fang2019} and entanglement distillation~\cite{Rains2001}.  
Let us now consider the task of entanglement dilution under LOCC operations, which is defined as above but with LOCC operations replacing PPT operations. We denote the corresponding figure of merit by $E^{\exact}_{ c,\,\locc}(\rho)\vphantom{\Big|}$. Due to PPT being an outer approximation to LOCC, it holds that
\bb
\ecost(\rho) \leq E^{\exact}_{ c,\,\locc}(\rho)
\ee
for all bipartite states $\rho = \rho_{AB}$. Our main result allows us to calculate the left-hand side, which is particularly significant because establishing computable lower bounds on $E^{\exact}_{ c,\,\locc}(\rho)$ is a priori difficult, as to do that one needs to 
constrain \emph{all} possible LOCC entanglement dilution protocols. (To prove an upper bound, on the contrary, it suffices to exhibit an explicit dilution protocol.)



We also remark that, in our definition, `zero error' means that we require the transformation of $\Phi_2^{\otimes \floor{Rn}}$ into $\rho^{\otimes n}$ to be realised exactly. This is a model of entanglement dilution that has been studied before in several different contexts~\cite{Martin-exact-PPT,HAYASHI,Buscemi2011}. The setting contrasts with definitions that allow for an asymptotically vanishing error in the transformation~\cite{Bennett-error-correction,Hayden-EC,Buscemi2011}
However, in the \tcb{Supplemental Material (SM)}~\cite{Note1} we show that no substantial change occurs if we require that the error, instead of being exactly zero, decay to zero sufficiently fast --- the relevant figure of merit is then still $\ecost$.

\medskip \noindent 
\textbf{\em Prior work.\,}---
In a pioneering paper by Audenaert, Plenio, and Eisert~\cite{Martin-exact-PPT} it was shown that the PPT entanglement cost $\ecost(\rho)$ can be evaluated exactly for all bipartite states $\rho = \rho_{AB}$ that satisfy a condition known as `zero bi-negativity'~\cite{Audenaert2002, Martin-exact-PPT, Ishizaka2004}. Specifically, if $|\rho^{\raisebox{-2pt}{\scriptsize $\Gamma$}}|^\Gamma \geq 0$, where $X^\Gamma$ denotes the partial transpose $\Gamma(X)$ of an operator $X$ and 
$|X|\coloneqq\sqrt{X^\dagger X}$ is 
its absolute value, then
\bb 
\ecost(\rho) = \log_2 \big\|\rho^\Gamma\big\|_1 \eqqcolon E_N(\rho)\, .
\label{APE_result}
\ee
The expression on the right, $E_N$, is 
the \emph{
logarithmic negativity}~\cite{negativity, plenioprl}, obtained by simply evaluating the trace norm $\|\cdot\|_1\coloneqq \Tr|\cdot|$ of the partially transposed state. This framework thus provides a partial operational interpretation for $E_N$ through its 
equality with the zero-error PPT cost $\ecost$ for some states. 
At the same time, Eq.~\eqref{APE_result} effectively solves the problem of computing $\ecost$ for states with zero bi-negativity, because $E_N$ is efficiently computable via a \emph{semi-definite program} (SDP)~\cite{WATROUS,Skrzypczyk_2023,vandenberghe_1996}:
\bb
E_N(\rho) = \log_2 \min\left\{ \Tr S:\ -S\leq \rho^\Gamma \leq S\right\} .
\ee

But what to do for those states that have non-zero bi-negativity, i.e.\ satisfy $|\rho^{\raisebox{-2pt}{\scriptsize $\Gamma$}}|^\Gamma \not\geq 0$? In 
their recent works~\cite{Xin-exact-PPT, Wang2023,Xin-Mark-errata}, Wang and Wilde showed that $\ecost$ can be expressed as a regularisation of another quantity called $E_\kappa$. They then made an even stronger claim that $E_\kappa$ is in fact \emph{additive}, meaning that $E_\kappa \big(\rho^{\otimes n}\big) \eqt{?} n\, E_\kappa(\rho)$ holds true for all bipartite states $\rho$ and all positive integers $n$, thus completely eliminating the issue of regularisation.
This would imply a general computable solution for 
$\ecost$: it would simply coincide with 
$E_\kappa$, which is computable via an SDP~\cite{Xin-exact-PPT, Wang2023,Xin-Mark-errata}.
However, in Lemma~\ref{punch_card_lemma} below we construct a simple counterexample that disproves the claimed additivity of $E_\kappa$. Its existence shows that in general $E_\kappa \neq \ecost$, thus invalidating the computable solution claimed in~\cite{Xin-exact-PPT, Wang2023} and reopening the question of whether the asymptotic cost $\ecost$ 
can be efficiently evaluated. Further details concerning the claims of~\cite{Xin-exact-PPT,Wang2023} can be found in~\cite{Xin-Mark-errata}.

The equivalence between the regularisation of the quantity $E_\kappa$ and the PPT entanglement cost revealed in~\cite{Xin-exact-PPT,Wang2023} will still prove useful to us, albeit \emph{a priori} it is not clear how it could lead to a computable formula for $\ecost$ --- the daunting problem of regularisation persists.

\medskip \noindent 
\textbf{\em Main results.\,}--- In this work, we completely solve the problem of computing the asymptotic zero-error PPT entanglement cost $\ecost$. 
To do this, we construct 
a converging hierarchy of semi-definite programs that can be used to calculate $\ecost$ for any given state to any degree of precision \emph{efficiently}, i.e.\ in time polynomial in the underlying Hilbert space dimension and in $\log(1/\e)$, with $\e$ being the additive error. 
The key quantities in our approach are a family of PPT entanglement monotones indexed by an integer $p\in \N$ and given by
\bb
E_{\chi, p}(\rho) &\coloneqq \log_2 \chi_p(\rho)\, ,
\label{E_chi_p}
\ee
where
\bb
\chi_p(\rho) &\!\coloneqq\! \min\! \Big\{ \Tr S_p\! :\, -S_i \leq S_{i-1}^\Gamma\! \!\leq\! S_i,\ i\!=\!0,...,p,\ S_{-1} \!=\! \rho \Big\}
\label{chi_p}
\ee
is an SDP with variables $S_0,\ldots, S_p$. These quantities are increasing in $p$ for every fixed $\rho$, and we refer to them as the \deff{$\boldsymbol{\chi}$-hierarchy}. Note also that $E_{\chi,0} = E_N$, hence $E_{\chi,p}$ can be regarded as a generalisation of the logarithmic negativity~\cite{negativity, plenioprl}.

Importantly, we show that the $\chi$-hierarchy approximates the entanglement cost $\ecost$ from below, in the sense that $\ecost(\rho) \geq E_{\chi, p}(\rho)$ for all states $\rho$ and all $p$. The proof of this fact, which can be found in the Supplemental Material~\cite{Note1}, relies on the connection between the entanglement cost $\ecost$ and the regularised form of $E_\kappa$ shown in~\cite{Xin-exact-PPT,Wang2023}. 
Our first main result, the forthcoming Theorem~\ref{convergence_thm}, establishes that this approximation becomes increasingly tight as $p$ increases, and the $\chi$-hierarchy gives the value of $\ecost$ \emph{exactly} in the limit $p \to\infty$. This allows us to replace the limit in the number of copies $n$, which is what makes $\ecost$ difficult to compute, with a limit in the hierarchy level $p$. This already provides a `single-letter' formula for the PPT entanglement cost
that no longer suffers from the curse of regularisation. However, because of the limiting procedure $p\to\infty$, it is still unclear if the expression can be evaluated easily. Crucially, in Theorem~\ref{convergence_thm} we also show that calculating the limit of the $\chi$-hierarchy is indeed significantly easier than evaluating regularised expressions: the convergence to the true value of $\ecost$ is \emph{exponentially fast} uniformly on all states, which opens the way to an accurate calculation of $\ecost$ in practice.

\begin{thm}[(Exact expression of the cost)]  \label{convergence_thm}
For all bipartite states $\rho=\rho_{AB}$ on a system of minimal local dimension $d\coloneqq \min\left\{|A|,|B|\right\}\geq 2$, and all positive integers $p\in \N^+$, it holds that
\bb
0\leq \ecost (\rho)- E_{\chi,p}(\rho)\leq \log_2 \frac{1}{1-\left(1 - \frac{2}{d}\right)^p}\, ,
\label{efficient_algorithm_key_inequality}
\ee
entailing that
\bb
E_{c,\,\ppt}^{\exact}(\rho)=\lim_{p\to\infty} E_{\chi,p}(\rho) \, .
\label{convergence_simplified}
\ee
\end{thm}
The proof of Theorem~\ref{convergence_thm} is outlined in the 
End Matter, with full details provided in the SM for interested readers~\cite{Note1}. Note that for every fixed value of $d$ and for large $p$, the approximation error on the right-hand side of~\eqref{efficient_algorithm_key_inequality} can be estimated as $\left(1 - \frac{2}{d}\right)^p (\log_2 e)$. In other words, the speed of convergence in~\eqref{convergence_simplified} is exponential in $p$ and furthermore independent of $\rho$.

The single-letter formula~\eqref{convergence_simplified} can be used to establish two notable properties of the zero-error PPT entanglement cost $\ecost$, namely additivity and continuity~\tcb{\cite[\S~\ref{subsec_additivity}--\ref{subsec_continuity}]{Note1}}. 
Also, Theorem~\ref{convergence_thm} yields immediately a 
simple solution in the qubit-qudit case ($d=2$), generalising Ishizaka's result that $\ecost = E_N$ for all two-qubit states~\cite{Ishizaka2004}:

\begin{cor}[(Simple formula for the cost of qubit-qudit states)] \label{qubit_case_cor}
For all states $\rho=\rho_{AB}$ on a $2\times n$ bipartite quantum system, it holds that
\bb
\ecost(\rho) &= E_{\chi,1}(\rho) \\
&= \log_2 \min\left\{ \big\| S_0^\Gamma\big\|_1:\, -S_0 \!\leq\! \rho^\Gamma\! \!\leq\! S_0 \right\}.
\label{qubit_case}
\ee
\end{cor}

But the most important implication of Theorem~\ref{convergence_thm} is that it allows us to construct an efficient algorithm that calculates $\ecost$ to any desired accuracy. This algorithm takes as input a bipartite state $\rho = \rho_{AB}$ and an error tolerance $\e>0$, and returns as output a number $\tilde{E}$ such that $\big|\ecost(\rho) - \tilde{E}\big| \leq \e$. It works as follows:
\begin{enumerate}
    \item Use~\eqref{efficient_algorithm_key_inequality} to find $p$ large enough so that $E_{\chi,p}(\rho)$ approximates $\ecost(\rho)$ up to an error $\epsilon/2$. A value of $p = \pazocal{O}\!\left(d \log (d/\epsilon)\right)$ is sufficient.
    \item Solve the SDP in~\eqref{chi_p} up to an (additive) error $\epsilon \ln(2)/2$. Taking the logarithm yields an estimate $\tilde{E}$ of $E_{\chi,p}(\rho)$ up to an (additive) error $\epsilon/2$ (see~\eqref{E_chi_p}).
    \item Return $\tilde{E}$.
\end{enumerate}



The time complexity of the above algorithm is analysed in the forthcoming Theorem~\ref{efficient_algorithm_thm}, which shows the core result of our work: the zero-error PPT entanglement cost can be efficiently computed. The key observation behind this result is that climbing the $\chi$-hierarchy up to level $p$ introduces only polynomially many more constraints in the SDP in~\eqref{E_chi_p} and is thus relatively inexpensive.

\begin{thm}[(Time to compute the cost)] \label{efficient_algorithm_thm}
Let $\rho$ be a bipartite state on a system of total dimension $D$ and minimal local dimension $d$. Then, 
the above algorithm computes $\ecost(\rho)$ up to an additive error $\e$ in time
\bb
\pazocal{O}\!\left((d D)^{6+o(1)} \operatorname{polylog}(1/\e) \right)\,.
\label{efficient_algorithm_time}
\ee
\end{thm}

The proof of Theorem~\ref{efficient_algorithm_thm} is sketched in the End Matter. Remarkably, Theorem~\ref{efficient_algorithm_thm} implies that the time complexity required to compute $\ecost$ is only marginally larger than that of computing the logarithmic negativity $E_N$, which is 
\tcb{bounded by the time complexity required for the diagonalisation of a $D\times D$ matrix --- also polynomial in $D$ and $\log(1/\e)$.} In short, while both $\ecost$ and $E_N$ can be computed efficiently, $\ecost$ stands out because, by definition, it has an operational meaning for all quantum states, unlike $E_N$. Our result is the first of its kind for two distinct reasons:

\begin{enumerate}[(a), wide, itemindent=4pt, topsep=3pt, itemsep=1pt]
\item First, because it establishes the efficient computability of an operationally meaningful asymptotic entanglement measure (i.e.\ a distillable entanglement or an entanglement cost). There is no known algorithm to estimate any other such measure, not even under the simplifying zero-error assumption.
\item Second, because efficient computability is shown without exhibiting a closed-form single-letter formula, 
but rather by describing a converging SDP hierarchy. 
To the extent of our knowledge, the only other case in quantum information theory where a similar situation arises is in~\cite[Theorem~5.1]{FawziFawzi}. However, unlike ours, the algorithm described there is computationally extremely expensive, featuring an exponential dependence on $d^3/\e$. More generally, expressing difficult-to-compute quantities through converging SDP hierarchies is a technical tool that has found various uses in quantum information~\cite{complete-extendibility, navascues_2008, harrow_2017-1, Fang2019, fawzi_2022, berta_2022-1}, but such 
applications typically do not result in efficiently computable algorithms or do not yield exact operational results.
\end{enumerate}

\medskip
\noindent \textbf{\em Discussion and conclusions.} In this paper, we have provided a 
solution to the problem of efficiently calculating the zero-error PPT entanglement cost of arbitrary (finite-dimensional) quantum states. To the best of our knowledge, it is the first time that any operational asymptotic entanglement measure is shown to be efficiently computable. A particularly interesting feature of our construction is that it does not rely on a closed-form formula, but rather on a converging hierarchy of semi-definite programs that approximate the cost from above and below with controllable error. 

Our solution identifies the correct generalisation of the celebrated logarithmic negativity, which became popular in quantum physics for its efficient computability, despite being operationally meaningful only for states with zero bi-negativity. In contrast, our entanglement measure, $\lim_{p\to\infty} E_{\chi,p}$, is not only efficiently computable --- specifically, it can be efficiently computed via 
an algorithm with almost the same time complexity as the logarithmic negativity --- but it also has operational meaning for \emph{all} quantum states, through its equality with the zero-error PPT entanglement cost $\ecost$. For states with zero bi-negativity, our entanglement measure coincides with the logarithmic negativity, but for 
other states it can be significantly smaller.

An open question in our analysis is whether the $\chi$-hierarchy collapses at any finite level for some --- or even all --- states. 
We found examples of states $\rho$ such that $E_N(\rho) = E_{\chi,0}(\rho) < E_{\chi,1}(\rho)$ and also $E_{\chi,1}(\rho) < E_{\chi,2}(\rho)$, but we were not able to 
ascertain whether there exists in general a gap between $E_{\chi,2}$ and $E_{\chi,3}$. If $E_{\chi,2} = E_{\chi,3}$ holds in general, then this would mean that $\ecost = E_{\chi,2}$, and thus the cost 
could be computed with a simple single-letter formula. While this would be a considerable simplification from the analytical standpoint, we stress that it will only entail a 
$\operatorname{poly}(d)$ improvement in the time complexity of evaluating it numerically.

\begin{acknowledgments}
\smallskip
\noindent \emph{Acknowledgements.} L.L.\ thanks Andreas Winter for enlightening correspondence on zero-error information theory, as well as the organisers and attendees of the workshop `Quantum information' (Les Diablerets, Switzerland, 25 February--1 March 2024), where this work was first presented, for stimulating and entertaining discussions. L.L.\ and B.R.\ thank Xin Wang and Mark M.\ Wilde for feedback on the first draft of this paper. \tcr{L.L.\ is supported by MIUR (Ministero dell'Istruzione, dell'Universit\`a e della Ricerca) through the project `Dipartimenti di Eccellenza 2023--2027' of the `Classe di Scienze' department at the Scuola Normale Superiore.}

\smallskip
\noindent \emph{Note.} The issue with the original argument by Wang and Wilde that leads to the additivity violation for $E_\kappa$ (Lemma~\ref{punch_card_lemma}) is discussed in detail in the erratum~\cite{Xin-Mark-errata}.
\end{acknowledgments}

\bibliographystyle{apsc}
\bibliography{biblio,appendix}

\section{End matter}
\medskip
\noindent \textbf{\em The task.}
We start by defining $\vphantom{\Big|}\ecost$ in rigorous terms. A (quantum) channel $\Lambda: X\to Y$ is a completely positive and trace preserving map taking as input states of a quantum system $X$ and outputting states of $Y$.  The set of completely positive maps (respectively, quantum channels) from $X$ to $Y$ will be denoted as $\mathrm{CP}(X\to Y)$ (respectively, $\cptp(X\to Y)$). If $X=AB$ and $Y=A'B'$ are both bipartite systems and $\Lambda\in \mathrm{CP}(AB\to A'B')$, we say that $\Lambda$ is \deff{PPT} if $\Gamma_{B'}\circ \Lambda \circ \Gamma_B$ is still completely positive, where $\Gamma_B$ denotes the partial transpose on $B$, and analogously for $\Gamma_{B'}$. 
Another way to understand this class of channels is to realise that they completely preserve the set of PPT states, in the sense that $\Gamma_{B_1 B'_2} \left[\operatorname{id} \otimes \Lambda(\sigma_{A_1 B_1 A_2 B_2})\right] \geq 0$ for any state $\sigma$ with $\Gamma_{B_1 B_2} (\sigma_{A_1 B_1 A_2 B_2}) \geq 0$. We can then define the \deff{zero-error PPT entanglement cost} of any bipartite state $\rho = \rho_{AB}$ as
\bb
\ecost(\rho) \coloneqq \inf\!\Big\{ &R\!:\, \text{for all sufficiently large}\ n\in \N \\ 
&\exists\ \Lambda_n \!\!\in\! \ppt \cap \cptp\!:\, \Lambda_n\big(\Phi_2^{\otimes \floor{Rn}}\big) \!=\! \rho^{\otimes n}
\Big\} .
\label{ecost}
\ee
Here, $\Phi_2 \coloneqq \ketbra{\Phi_2}$, where $\ket{\Phi_2} \coloneqq \frac{1}{\sqrt2} \left(\ket{00} + \ket{11}\right)$ is the two-qubit maximally entangled state, i.e.\ the \deff{ebit}, and $\Lambda_n$ is required to be a PPT channel.


\smallskip
\noindent 
\textbf{\em The quantifier $\boldsymbol{E_\kappa}$.} Wang and Wilde~\cite{Xin-exact-PPT, Wang2023, Xin-Mark-errata} introduced and studied the SDP-computable quantity
\bb
E_\kappa(\rho) \coloneqq \log_2 \min\left\{ \Tr S:\, -S\leq \rho^\Gamma\leq S,\ S^\Gamma\geq 0\right\} .
\label{E_kappa}
\ee
Among other things, they showed that: (i)~$E_\kappa$ is monotonically non-increasing under PPT channels; (ii)~$E_\kappa(\rho)\geq E_N(\rho)$ for all states $\rho$, with equality when $\rho$ has zero bi-negativity; 
(iii)~$E_\kappa$ is sub-additive,  meaning that
\bb
E_\kappa(\rho\otimes \rho') \leq E_\kappa(\rho) + E_\kappa(\rho')
\label{subadditivity_E_kappa}
\ee
for all pairs of states $\rho = \rho_{AB}$ and $\rho' = \rho'_{A'B'}$; and (iv)~its regularisation yields the zero-error PPT entanglement cost, i.e.
\bb\label{E_k_inf}
\ecost(\rho) = E_\kappa^\infty(\rho) \coloneqq \lim_{n\to\infty} \frac1n\, E_\kappa\big(\rho^{\otimes n}\big)\, .
\ee
As said, it was claimed in~\cite{Xin-exact-PPT} that $E_\kappa$ is additive, meaning that equality holds in~\eqref{subadditivity_E_kappa}. However, this claim is incorrect (see also~\cite{Xin-Mark-errata}). To disprove it, it is useful to first note that additivity indeed holds when both $\rho$ and $\rho'$ have zero bi-negativity, simply because in that case $E_\kappa$ coincides with the logarithmic negativity by property~(ii), and this latter measure \emph{is} additive, as one sees immediately by looking at its definition in terms of the $1$-norm of the partially transposed state. Hence, our search for a counterexample must start with the construction of states with non-zero bi-negativity. That such states do exist was reported already in~\cite{Audenaert2002, Ishizaka2004} based on numerical evidence. However, here we present a simpler, analytical construction. 

\smallskip
\noindent 
\textbf{\em Punch card states.}
Let $A\geq 0$ be a positive semi-definite $d\times d$ matrix, and let $Q$ be another $d\times d$ symmetric matrix with only $0$/$1$ entries and all $1$'s on the main diagonal (i.e.\ such that $Q_{ii}=1$ for all $i$). The associated \deff{punch card state} is the bipartite quantum state on $\C^d\otimes \C^d$ defined by
\bb
\pi_{A,Q} \coloneqq \frac{1}{N_{A,Q}}\left( \sum_{i,j} A_{ij} \ketbraa{ii}{jj} + \sum_{i\neq j} Q_{ij} |A_{ij}| \ketbra{ij}\right) ,
\label{punch_card_state}
\ee
where $N_{A,Q}$ is chosen so that $\Tr \pi_{A,Q} = 1$. 
It can be verified that
\bb
 \pi_{A,Q}^\Gamma \!&\simeq \frac{1}{N_{A,Q}} \left[ \left(\sumno_i \!A_{ii} \ketbra{ii} \right) \oplus \bigoplus_{i<j} \begin{pmatrix} Q_{ij} |A_{ij}| & A_{ij} \\ A_{ij}^* & Q_{ij} |A_{ij}| \end{pmatrix} \right]\! ,
\ee
from which, using the identity
\bb
\left| \begin{pmatrix} Q_{ij} |A_{ij}| & A_{ij} \\ A_{ij}^* & Q_{ij} |A_{ij}| \end{pmatrix} \right| = \begin{pmatrix} |A_{ij}| & Q_{ij} A_{ij} \\ Q_{ij} A_{ij}^* & |A_{ij}| \end{pmatrix} ,
\ee
valid because $Q_{ij}\in\{0,1\}$, 
one derives that
\bb
 \big|\pi_{A,Q}^\Gamma \big|^\Gamma = \frac{1}{N_{A,Q}} \left(\sum_{i,j} Q_{ij} A_{ij} \ketbraa{ii}{jj} + \sum_{i\neq j} |A_{ij}| \ketbra{ij}\right) .
 \label{transformed_punch_card_state}
\ee
Therefore, if $A$ and $Q$ are chosen such that $Q\circ A\not\geq 0$, where $\circ$ denotes the Hadamard (i.e.\ entry-wise) product between matrices, $\pi_{A,Q}$ will have non-zero bi-negativity, i.e.\ $\big|\pi_{\text{\raisebox{3pt}{\scriptsize $A$,$Q$}}}^\Gamma \big|^\Gamma \not\geq 0$.
It is easy to construct examples of $A$ and $Q$ that meet the above criteria, the simplest one being
\bb
A_0 \coloneqq \begin{pmatrix} 1 & 1 & 1 \\ 1 & 1 & 1 \\ 1 & 1 & 1 \end{pmatrix} ,\qquad Q_0 \coloneqq \begin{pmatrix}  1 & 0 & 1 \\ 0 & 1 & 1 \\ 1 & 1 & 1 \end{pmatrix} .
\label{A_0_Q_0}
\ee
Having constructed a state with non-zero bi-negativity, we can wonder whether two copies of it already violate the additivity of $E_\kappa$. And sure enough, they do:

\begin{lemma} \label{punch_card_lemma}
The punch card state $\pi_0 \coloneqq \pi_{A_0,Q_0}$ defined by~\eqref{punch_card_state} with the substitution~\eqref{A_0_Q_0} satisfies 
\bb
1.001 \approx E_\kappa\big(\pi_0^{\otimes 2}\big) < 2 E_\kappa(\pi_0) \approx 1.029\, .
\ee
\end{lemma}
In particular, the sub-additivity of $E_\kappa$ and Lemma~\ref{punch_card_lemma} imply 
that $\ecost(\pi_0) \leq \frac12\, E_\kappa\big(\pi_0^{\otimes 2}\big) < E_\kappa(\pi_0)$ and therefore that $\ecost(\pi_0) \ne 
E_\kappa(\pi_0)$, thus invalidating the main claim of the works~\cite{Xin-exact-PPT,Wang2023}.

\smallskip
\noindent 
\textbf{\em Two SDP hierarchies.} Recall that in~\eqref{E_chi_p}--\eqref{chi_p} we introduced the $\chi$-hierarchy $E_{\chi,p} (\rho) = \log_2 \chi_p(\rho)$ with
\bb
\label{chi_p_redef}
\chi_p(\rho) &\coloneqq \min \Big\{ \Tr S_p\! :\, -S_i \!\leq\! S_{i-1}^\Gamma\! \!\leq\! S_i,\ i\!=\!0,...,p,\ S_{-1} \!=\! \rho \Big\},
\ee
which constitutes a generalisation of the logarithmic negativity. We now introduce another complementary hierarchy of SDPs, the \deff{$\boldsymbol{\kappa}$-hierarchy}, defined for $q\in \N^+$ ($q\geq 1$) by
\bb
E_{\kappa, q}(\rho) \coloneqq \log_2 \kappa_q(\rho)\, ,
\label{E_kappa_p}
\ee
with
\bb
\kappa_q(\rho) \!\coloneqq\! \min\! \Big\{ \!\Tr S_{q-1}\!\! : -S_i \!\leq\! S_{i-1}^\Gamma\! \!\leq\! S_i,\ \, &i\!=\!0,\!...,q\!-\!1, \\
&\hspace{-3.3ex} S_{-1} \!=\! \rho,\ S_{q-1}^\Gamma \!\geq\! 0 \Big\}\! .
\label{kappa_p}
\ee
Observe the resemblance to the definition of $\chi_p$ in~\eqref{chi_p_redef}: the only difference between the two optimisations is the condition $S_{p-1}^\Gamma \geq 0$; adding that to~\eqref{chi_p_redef} yields immediately~\eqref{kappa_p} with $q\mapsto p$, and indeed in that case the optimal $S_p$ would automatically be $S_p = S_{p-1}^\Gamma$. Furthermore, $E_{\kappa, 1} = E_\kappa$ coincides with the quantity~\eqref{E_kappa} introduced by Wang and Wilde, of which the $\kappa$-hierarchy thus constitutes a generalisation.

In the SM~\cite{Note1} we explore the properties of the quantities $E_{\chi,p}$ and $E_{\kappa,q}$, showing them to be 
legitimate entanglement measures. In particular, the functions are all suitably normalised, continuous, faithful on PPT states, and strongly monotonic under PPT operations. 
The pivotal property that distinguishes $E_{\chi,p}$ from $E_{\kappa,q}$ is that, while the quantities $E_{\kappa,q}$ are only sub-additive, the $\chi$-quantities are fully additive under tensor products, meaning that regularisation can always be avoided.



Two key insights lead to the proof of our main results. First, 
that the two hierarchies provide complementary bounds on $\ecost$, the fundamental quantity we want to estimate: namely, the $\chi$-hierarchy gives increasing lower bounds on $\ecost$, while the $\kappa$-hierarchy gives decreasing upper bounds on it. In other words, on any fixed state $\rho$
\bb
E_{\chi,0} \leq E_{\chi,1} \leq ... \leq \ecost = E_\kappa^\infty \leq ... \leq E_{\kappa,2} \leq E_{\kappa,1} .
\label{hierarchy}
\ee
In particular, $E_{\chi,p}$ is increasing in $p$, while $E_{\kappa,q}$ is decreasing in $q$. 
Eq.~\eqref{hierarchy} 
immediately shows the remarkable connection between two very different limits: one in the number of copies $n$, which is needed to compute $E_\kappa^\infty$ (see Eq.~\eqref{E_k_inf}), and one in the hierarchy levels $p$ and $q$.

The second insight is that there is a connection between the $\chi$- and $\kappa$-hierarchies, as expressed by the following key technical result, proven in the SM~\cite{Note1}.

\begin{prop} \label{convergence_primitive_prop}
For all states $\rho=\rho_{AB}$ on a system of minimal local dimension $d \coloneqq \min\{|A|,|B|\} \geq 2$, and all $p\in \N^+$,
\bb
\kappa_p(\rho) \leq \frac{d}{2}\, \chi_p(\rho) - \left(\frac{d}{2}-1\right) \chi_{p-1}(\rho)\, .
\label{convergence_primitive}
\ee
\end{prop}

With Proposition~\ref{convergence_primitive_prop} at hand, we can now see how it implies our two main results, Theorems~\ref{convergence_thm} and~\ref{efficient_algorithm_thm}.

\begin{proof}[Proof sketch of Theorem~\ref{convergence_thm}]
Combining~\eqref{hierarchy} and~\eqref{convergence_primitive} shows that
\bb
2^{\ecost(\rho)} \leq \kappa_p(\rho) \leq \frac{d}{2}\, \chi_p(\rho) - \left(\frac{d}{2}-1\right) \chi_{p-1}(\rho)\, .
\label{convergence_thm_proof_eq1}
\ee
The quantity that we are really interested in, however, is the normalised difference between $\chi_p(\rho)$ and its claimed limiting value $2^{\ecost(\rho)}$. To see what the above inequality tells us in this respect, we can define the quantity
\bb
\e_p(\rho) \coloneqq 1 - \frac{\chi_p(\rho)}{2^{\ecost(\rho)}}\, ,
\label{convergence_thm_proof_eq2}
\ee
by means of which~\eqref{convergence_thm_proof_eq1} can be cast as $\e_p(\rho) \leq \left(1-\frac2d\right) \e_{p-1}(\rho)$. 
Iterating the above relation gives
\bb
\e_p(\rho) \leq \left(1-\frac2d\right)^p \!\e_0(\rho) \leq \left(1-\frac2d\right)^p\, ,
\ee
which entails~\eqref{efficient_algorithm_key_inequality} after elementary algebraic manipulations. Taking the limit $p\to\infty$ in~\eqref{efficient_algorithm_key_inequality} proves also~\eqref{convergence_simplified}.
\end{proof}

\begin{proof}[Proof sketch of Theorem~\ref{efficient_algorithm_thm}]
It suffices to formalise the qualitative argument provided below the statement of the theorem. For $d>2$, we first choose
\bb
p_d \coloneqq \ceil{\frac{\log_2 (2d/\e)}{- \log_2 \left(1-\frac2d\right)}} = \pazocal{O}\!\left(d 
\log (d/\e)\right) ,
\label{efficient_algorithm_p}
\ee
so that, with the notation in~\eqref{convergence_thm_proof_eq2}, $\e_{p_d}(\rho)\leq \frac{\e}{2d}$. 
Using 
$\ln(2)\, |a-b|\leq \left|2^a - 2^b\right|$, valid for $a,b\geq 0$, we get
\bb
0 \leq \ecost(\rho) - E_{\chi,p_d}(\rho) \leq \frac{2^{\ecost(\rho)} \e}{2d\ln 2}\le \frac{\e}{2\ln 2}\, ,
\label{efficient_algorithm_first_approx}
\ee
where the last inequality is a consequence of the fact that every state can be created via a quantum teleportation protocol --- and hence with PPT operations --- from a maximally entangled state, which entails that $E_{c,\,\ppt}^\exact(\rho)\leq \log_2 d$ for all $\rho$. We then solve the SDP for $\chi_{p_d}(\rho)$ up to an additive error $\left(\ln 2 - 1/2\right)\e$ by running an optimised SDP solver~\cite{LSW2015, vanApeldoorn2020}. Doing so yields an approximation of $E_{\chi,p_d}(\rho)$ up to an additive error $\left(1-\frac{1}{2\ln 2}\right)\e$. Adding this up with the error in~\eqref{efficient_algorithm_first_approx} yields a total error of $\e$. The time complexity in~\eqref{efficient_algorithm_time} can be calculated using 
known theoretical bounds on the complexity of SDPs, e.g.\ those found in~\cite{LSW2015}.
%
\end{proof}

\let\addcontentsline\oldaddcontentsline

\clearpage
\fakepart{Supplemental Material}

\onecolumngrid
\begin{center}
\vspace*{\baselineskip}
{\textbf{\large Supplemental Material}}\\
\end{center}

\renewcommand{\theequation}{S\arabic{equation}}
\renewcommand{\thethm}{S\arabic{thm}}
\renewcommand{\thefigure}{S\arabic{figure}}
\setcounter{page}{1}
\makeatletter

\setcounter{secnumdepth}{2}

\tableofcontents

\section{Introduction} \label{sec_intro}

\subsection{Quantum states and channels}

The quantum systems we consider in this work are all represented by finite-dimensional Hilbert spaces $\HH$. We will denote the set of linear operators on $\HH$ by $\LL (\HH)$. We occasionally refer to these operators also as matrices, to underline the fact that they are understood to be finite dimensional.

Quantum states pertaining to a quantum system with Hilbert space $\HH$ are represented by \deff{density operators}, i.e.\ positive semi-definite operators on $\HH$ with unit trace. We will denote the set of positive semi-definite operators on $\HH$ by $\LL_+(\HH)$. Given two Hermitian operators $X,Y\in \LL_+(\HH)$ we write $X\leq Y$ if $Y-X \in \LL_+(\HH)$, and $X<Y$ if $Y-X$ is strictly positive definite. The relations $\geq$ and $>$ are defined similarly.

The \deff{trace norm} of an arbitrary operator $X\in \LL(\HH)$ is given by 
\bb
\|X\|_1 \coloneqq \Tr \sqrt{X^\dag X}\, ,
\ee
where for a positive semi-definite operator $A$ its square root is constructed as the unique semi-definite operator $\sqrt{A}$ such that $\sqrt{A}^2 = A$. If $X=X^\dag$ is Hermitian, then it is possible to give the following simple characterisation of its trace norm. We include a proof for the sake of completeness, as the lemma below will be used a few times throughout this paper.

\begin{lemma}[(Variational characterisation of the trace norm)] \label{variational_trace_norm_lemma}
Let $X=X^\dag\in \LL(\HH)$ be a Hermitian matrix. Then
\bb
\|X\|_1 = \min\left\{ \Tr Y:\ -Y\leq X\leq Y\right\} .
\ee
\end{lemma}

\begin{proof}
Let $X = \sum_i x_i \ketbra{i}$ be the spectral decomposition of $X$. Then by inspection we find $\|X\|_1 = \sum_i |x_i|$. On the one hand, setting $Y = |X| \coloneqq \sum_i |x_i| \ketbra{i}$ we find that $\min\left\{ \Tr Y:\ -Y\leq X\leq Y\right\}\leq \Tr \sum_i |x_i| \ketbra{i} = \|X\|_1$. On the other, for an arbitrary $Y$ such that $-Y\leq X\leq Y$ we have that $-\braket{i|Y|i}\leq \braket{i|X|i} = x_i \leq \braket{i|Y|i}$, implying that $\braket{i|Y|i} \geq |x_i|$. Therefore, $\Tr Y = \sum_i \braket{i|Y|i} \geq \sum_i |x_i| = \|X\|_1$; since $Y$ was arbitrary, this concludes the proof.
\end{proof}

Given two quantum systems with Hilbert spaces $\HH$ and $\HH'$, a linear map $\Lambda:\LL(\HH)\to \LL(\HH')$ is called \deff{positive} if $\Lambda\big( \LL_+(\HH)\big)\subseteq \LL_+(\HH')$, and \deff{completely positive} if $I_k \otimes \Lambda$ is positive for all $k\in \N$, where $I_k$ is the identity map on $\LL\big(\C^k\big)$. It is well known that complete positivity can be formulated as a positive semi-definite constraint using the formalism of Choi states. Namely, introducing the maximally entangled state
\bb
\Phi_d \coloneqq \ketbra{\Phi_d}\, ,\qquad \ket{\Phi_d}\coloneqq \frac{1}{\sqrt{d}}\sum_{i=1}^d \ket{ii}\, ,
\label{maximally_entangled_state}
\ee
a linear map $\Lambda:\LL\big(\C^d\big)\to \LL\big(\C^{d'}\big)$ is completely positive if and only if
\bb
\big(I_d\otimes \Lambda\big)(\Phi_d)\geq 0
\label{CP_condition}
\ee
is a positive semi-definite operator.

A completely positive map and trace-preserving map $\Lambda:\LL(\HH)\to \LL(\HH')$ is called a \deff{quantum channel}. In what follows, $\mathrm{CP}(X\to Y)$ and $\cptp(X\to Y)$ will denote respectively the set of completely positive maps and the set of quantum channels between two quantum systems $X$ and $Y$.

\subsection{Positive partial transpose}

Composite quantum systems are represented by the tensor product of the local Hilbert spaces; for the simplest case of a bipartite system, this is expressed in formula as $\HH_{AB} = \HH_A \otimes \HH_B$. In what follows, we will often need to consider the minimal local dimension of a bipartite quantum system $AB$, i.e.\ the number $d = \min\{|A|,|B|\}$, where we denote by $|A|$ (respectively, $|B|$) the dimension of the local Hilbert space $\HH_A$ (respectively, $\HH_B$).

Given an arbitrary bipartite quantum system $AB$, the \deff{partial transpose} on $B$ is the linear map $\Gamma_B : \LL\big(\HH_{AB}\big) \to \big(\HH_{AB}\big)$ uniquely defined by
\bb
\Gamma_B(X_A\otimes Y_B) = (X_A\otimes Y_B)^\Gamma \coloneqq X_A \otimes Y_B^\intercal\qquad \forall\ X_A\in \LL(\HH_A)\, ,\ Y_B\in \LL(\HH_B)\, .
\ee
Here, $\intercal$ denotes the transposition with respect to a fixed basis of $\HH_B$. Note that
\bb
\Tr Z_{AB} W_{AB} = \Tr Z_{AB}^\Gamma W_{AB}^\Gamma \qquad \forall\ Z,W\in \LL(\HH_{AB})\, ,
\label{partial_trace_HS_preserving}
\ee
and in particular
\bb
\Tr Z_{AB}^\Gamma = \Tr Z_{AB}\qquad \forall\ Z\in \LL(\HH_{AB})\, .
\label{partial_trace_trace_preserving}
\ee

The definition of $\Gamma_B$ does depend on which basis we choose, but all choices are unitarily equivalent to each other, in the sense that for any two different partial transpositions $\Gamma,\,\widetilde{\Gamma}$ there exists a unitary $U_B$ acting on $B$ such that
\bb
\widetilde{\Gamma}_B(Z_{AB}) = U_B^{\vphantom{\dag}}\, \Gamma_B(Z_{AB})\, U_B^\dag .
\ee

A state $\sigma_{AB}$ on a bipartite quantum system $AB$ is said to be a \deff{PPT state} if $\sigma_{AB}^\Gamma = \Gamma_B(\sigma_{AB})\geq 0$. The PPT condition is important in quantum information because it provides a relaxation of separability: any \deff{separable} (i.e.\ unentangled) state $\sigma_{AB}$, which can thus be decomposed as $\sigma_{AB} = \sum_x p_x\, \alpha_x^A \otimes \beta_x^B$ for some probability distribution $p$ and some sets of local states $\alpha_x^A$ and $\beta_x^B$~\cite{Werner}, is PPT, but the converse is famously not true~\cite{HorodeckiBound, Bruss2000}.

A completely positive map $\Lambda\in \mathrm{CP}(AB\to A'B')$ with bipartite quantum systems $AB$ and $A'B'$ as input and output is called \deff{PPT} if
\bb
\Gamma_{B'} \circ \Lambda \circ \Gamma_B \in \mathrm{CP}(AB\to A'B')\, .
\label{PPT_map}
\ee
If $\Lambda$ is in addition also trace preserving, then we call it a \deff{PPT channel} (or a \deff{PPT operation}). Such maps are sometimes also referred to as `completely PPT preserving'. The importance of the set of PPT channels is that it provides an outer approximation to the set of quantum channels that can be implemented with local operations and classical communication (LOCC) on a bipartite system $AB$. The usefulness of the PPT approximation rests on the key observation that condition~\eqref{PPT_map} amounts to a positive semi-definite constraint, as can be seen by employing~\eqref{CP_condition}.

\subsection{Zero-error PPT entanglement cost}

The \deff{zero-error PPT entanglement cost} of a bipartite state $\rho=\rho_{AB}$ is defined as the minimum rate of \deff{ebits} (or singlets)
\bb
\Phi_2\coloneqq \ketbra{\Phi_2}\, ,\qquad \ket{\Phi_2}\coloneqq \frac{1}{\sqrt2}\left(\ket{00} + \ket{11}\right) ,
\ee
that need to be consumed in order to create copies of $\rho$ with PPT operations and zero error. In formula,
\bb
\ecost(\rho) \coloneqq \inf\left\{ R:\ \text{for all sufficiently large}\ n\in \N\ \exists\ \Lambda_n \in \ppt \cap \cptp:\ \Lambda_n\big(\Phi_2^{\otimes \floor{Rn}}\big) \!=\! \rho^{\otimes n} \right\} .
\label{ecost_SM}
\ee
Wang and Wilde proved~\cite{Xin-exact-PPT,Xin-Mark-errata} that $\ecost$ can be alternatively computed as\footnote{\tcr{Note that $E_\kappa$ is clearly finite for all finite-dimensional states, as one deduces by making the ansatz $S=\id$ in the optimisation~\eqref{E_kappa_SM}. In fact, using the convexity of the semi-definite program inside the logarithm in~\eqref{E_kappa_SM} together with the Schdmit decomposition for pure states on a bipartite system, it is not difficult to show that $E_\kappa$ is finite provided that $\min\{|A|,|B|\}<\infty$.}}
\begin{align}
\ecost(\rho) =&\ E_\kappa^\infty(\rho) \coloneqq \lim_{n\to\infty} \frac1n\, E_\kappa\big(\rho^{\otimes n}\big)\, , \label{cost_E_kappa_regularised} \\
E_\kappa(\rho) \coloneqq&\ \log_2 \min\left\{ \Tr S:\ -S\leq \rho^\Gamma\! \leq S,\ S^\Gamma \geq 0 \right\} . \label{E_kappa_SM}
\end{align}
The first equality in~\eqref{cost_E_kappa_regularised} is an asymptotic consequence of the more general one-shot statement that the PPT entanglement cost of generating a single copy of a state $\rho$ with PPT operations, denoted by $E_{c,\,\ppt}^{(1)}(\rho)$, can be bounded as $\log_2\big( 2^{E_\kappa(\rho)} - 1 \big) \leq E_{c,\,\ppt}^{(1)}(\rho) \leq \log_2\big( 2^{E_\kappa(\rho)} + 2 \big)$~\cite[Proposition~2]{Xin-exact-PPT}. It is easy to see that the function $E_\kappa$ is sub-additive, so that by Fekete's lemma~\cite{Fekete1923} we can alternatively write the limit in~\eqref{cost_E_kappa_regularised} as an infimum over $n$, i.e.
\bb
E_\kappa^\infty(\rho) = \inf_{n\in \N^+} \frac1n\, E_\kappa\big(\rho^{\otimes n}\big)\, .
\label{E_kappa_regularised_infimum}
\ee
Unfortunately, due to an error in the derivation of the super-additivity inequality, which can be traced back to~\cite[Eq.~(S59)]{Xin-exact-PPT}, $E_\kappa$ is no longer known to be additive. Hence, the regularisation in the above identity is needed. And in fact, we will show in the forthcoming Section~\ref{sec_additivity_violations} that it \emph{must} be included, because $E_\kappa$ is indeed \emph{not} additive in general.

We refer the interested reader to the erratum~\cite{Xin-Mark-errata} for a detailed explanation of this issue and of its consequences. There, Wang and Wilde also present other examples of additivity violations for $E_\kappa$. In the next section, we will construct our own examples of additivity violation.

\section{Violations of additivity of \texorpdfstring{$E_\kappa$}{E\_kappa}} \label{sec_additivity_violations}

To construct an example of additivity violation, it is useful to start by asking the opposite question: when is it that we know $E_\kappa(\rho)$ to be additive? There is a simple answer to this question that is already discussed in~\cite{Xin-exact-PPT} based on the pioneering work~\cite{Martin-exact-PPT}: $E_\kappa$ is easily seen to be additive on many copies of $\rho$ whenever $\rho$ has zero bi-negativity, i.e.\ whenever $\big|\rho^\Gamma\big|^\Gamma\geq 0$. When this happens, then
\bb
E_\kappa\big(\rho^{\otimes n}\big) = \log_2 \left\| \big(\rho^{\otimes n}\big)^\Gamma\right\|_1 = n \log_2 \left\|\rho^\Gamma\right\|_1 = n E_N(\rho) = n E_\kappa(\rho)\, ,
\ee
for all $n$. This observation tells us that to look for additivity violations of $E_\kappa$ we need to construct states with non-zero bi-negativity. 
After a bit of trial and error, we homed in on the following construction.

\begin{Def}[(Punch card states)]
Let $A\geq 0$ be a $d\times d$ positive semi-definite matrix. Also, let $Q$ be a $d\times d$ symmetric matrix with $Q_{ij}\in \{0,1\}$ for all $i,j=1,\ldots, d$, and $Q_{ii}=1$ for all $i=1,\ldots,d$. Then the associated punch card state is a bipartite quantum state acting on $\C^d\otimes \C^d$ and defined by
\bb
\pi_{A,Q} \propto \sum_{i,j} A_{ij} \ketbraa{ii}{jj} + \sum_{i\neq j} Q_{ij} |A_{ij}| \ketbra{ij}\, ,
\ee
with the normalisation $\Tr \pi_{A,Q} = 1$.
\end{Def}

Note that
\bb
\pi_{A,Q}^\Gamma &\propto \sum_i A_{ii} \ketbra{ii} + \sum_{i<j} \left( Q_{ij} |A_{ij}| \left(\ketbra{ij} + \ketbra{ji} \right) + A_{ij} \ketbraa{ij}{ji} + A_{ij}^* \ketbraa{ji}{ij} \right) \\
&\simeq \left(\sumno_i A_{ii} \ketbra{ii} \right) \oplus \bigoplus_{i<j} \begin{pmatrix} Q_{ij} |A_{ij}| & A_{ij} \\ A_{ij}^* & Q_{ij} |A_{ij}| \end{pmatrix} ,
\ee
from which, using the elementary observation that
\bb
\left| \begin{pmatrix} Q_{ij} |A_{ij}| & A_{ij} \\ A_{ij}^* & Q_{ij} |A_{ij}| \end{pmatrix} \right| = \begin{pmatrix} |A_{ij}| & Q_{ij} A_{ij} \\ Q_{ij} A_{ij}^* & |A_{ij}| \end{pmatrix} ,
\ee
valid because $Q_{ij}\in\{0,1\}$, we obtain that
\bb
\left|\pi_{A,Q}^\Gamma \right|^\Gamma \propto \sum_{i,j} Q_{ij} A_{ij} \ketbraa{ii}{jj} + \sum_{i\neq j} |A_{ij}| \ketbra{ij}\, ,
\ee
where $\propto$ hides a positive constant. We therefore record the following observation:

\begin{lemma}
Let $A\geq 0$ be a $d\times d$ positive semi-definite matrix, and let $Q$ be a $d\times d$ symmetric matrix with $Q_{ij}\in \{0,1\}$ for all $i,j=1,\ldots, d$, and $Q_{ii}=1$ for all $i=1,\ldots,d$. If $Q\circ A \not\geq 0$, where $\circ$ denotes the Hadamard product, then the punch card state $\pi_{A,Q}$ does not have zero bi-negativity, meaning that
\bb
\left|\pi_{A,Q}^\Gamma \right|^\Gamma \not\geq 0\, .
\ee
In particular, the state
\bb
\pi_0 \coloneqq \pi_{A_0,Q_0}\, ,\qquad A_0 \coloneqq \begin{pmatrix} 1 & 1 & 1 \\ 1 & 1 & 1 \\ 1 & 1 & 1 \end{pmatrix} ,\qquad Q_0 \coloneqq \begin{pmatrix}  1 & 0 & 1 \\ 0 & 1 & 1 \\ 1 & 1 & 1 \end{pmatrix} \not\geq 0
\label{pi_zero_SM}
\ee
does not have zero bi-negativity.
\end{lemma}

And sure enough, the above punch card state is already enough to exhibit a small but meaningful violation of additivity for $E_\kappa$, already at the $2$-copy level:

\begin{prop}
For the state $\pi_0$ defined by~\eqref{pi_zero_SM}, it holds that
\bb
2 \ecost(\pi_0)\leq E_\kappa\big(\pi_0^{\otimes 2}\big) < 2 E_\kappa(\pi_0)\, .
\ee
In particular, $E_\kappa\big(\pi_0^{\otimes 2}\big)\approx 1.001$ and $2E_\kappa(\pi_0) \approx 1.029$.
\end{prop}

\begin{proof}
The first inequality follows by combining~\eqref{cost_E_kappa_regularised} and~\eqref{E_kappa_regularised_infimum}. The numerical calculations of $E_\kappa\big(\pi_0^{\otimes 2}\big)$ and $E_\kappa(\pi_0)$ can be verified with an SDP solver.
\end{proof}

\section{New entanglement monotones: \texorpdfstring{$E_{\chi,{\MakeLowercase p}}$}{E\_chi,p} and \texorpdfstring{$E_{\kappa,{\MakeLowercase q}}$}{E\_kappa,q}}

\subsection{Rationale of our construction}

Let us investigate in greater depth the fundamental reason why $E_\kappa$ fails to be additive. To this end, it is useful to look at the dual SDP program~\cite{Skrzypczyk_2023,vandenberghe_1996} for $E_\kappa$, which takes the form~\cite[Eq.~(S25)]{Xin-exact-PPT}
\bb
E_\kappa(\rho) &= \log_2 \max \left\{ \Tr \rho\, (V - W)^\Gamma:\ V,W\geq 0,\ V^\Gamma + W^\Gamma \leq \id \right\} .
\label{E_kappa_dual_SM}
\ee
Since $E_\kappa$ is sub-additive, to prove additivity we would only need to prove super-additivity. This is most naturally done with the dual program, which involves a maximisation. The proof of $E_\kappa(\rho\otimes \rho') \geqt{?} E_\kappa(\rho) + E_\kappa(\rho')$ would proceed as follows: first we consider optimal feasible points $V,W$ and $V',W'$ that achieve the maxima in the dual SDPs~\eqref{E_kappa_dual_SM} for $\rho$ and $\rho'$; then we attempt to construct a feasible point $V'',W''$ for the dual SDP for $\rho\otimes \rho'$; such an ansatz would need to yield a value of the objective function equal to $E_\kappa(\rho) + E_\kappa(\rho')$ (upon taking the logarithm). A quick inspection reveals that this requires that
\bb
\Tr \rho\, (V''- W'')^\Gamma \eqt{?} \left( \Tr \rho\, (V - W)^\Gamma\right) \left( \Tr \rho\, (V' - W')^\Gamma \right) .
\label{ansatz_condition}
\ee
A natural ansatz at this point is the one found by Wang and Wilde~\cite{Xin-exact-PPT}, namely,
\bb
V'' = V\otimes V' + W\otimes W'\, ,\qquad W'' = V\otimes W' + W\otimes V'\, .
\ee
Note that $V'',W''\geq 0$; furthermore, this ansatz clearly satisfies~\eqref{ansatz_condition}. The problem is that
\bb
(V''+W'')^\Gamma = (V+W)^\Gamma \otimes (V' + W')^\Gamma \not\leq \id
\ee
in general. Indeed, while $(V+W)^\Gamma \leq \id$ and $(V'+W')^\Gamma \leq \id$ holds by hypothesis, from this it does not follow that the tensor product of these operators is also upper bounded by the identity. The reason is that both of these operators could have large negative eigenvalues, which would become large and positive when multiplied.

The first idea we have is to modify the definition of $E_\kappa$ so as to make sure that this does not happen. To this end, it suffices to add the further constraint $(V+W)^\Gamma \geq -\id$. We can therefore define
\bb
E_\chi(\rho) \coloneqq&\ \log_2 \max \left\{ \Tr \rho\, (V - W)^\Gamma:\ V,W\geq 0,\ -\id\leq V^\Gamma + W^\Gamma \leq \id \right\} \\
=&\ \log_2 \min\left\{ \Tr T :\ -S\leq \rho^\Gamma\! \leq S,\ -T \leq S^\Gamma\! \leq T \right\} \\
=&\ \log_2 \min\left\{ \left\|S^\Gamma \right\|_1 :\ -S\leq \rho^\Gamma\! \leq S \right\} ,
\label{E_chi_SM}
\ee
where in the second line we derived the primal program corresponding to the dual program in the first line, and in the third line we used Lemma~\ref{variational_trace_norm_lemma}. 

We will see shortly that our approach fixes the additivity problem, in the sense that the new function $E_\chi$ we have just constructed is indeed additive, unlike $E_\kappa$. Of course, what we are really interested in is the regularisation of $E_\kappa$; in general, $E_\chi$ can only provide a lower bound to that quantity. Indeed, for all states it holds that $E_\chi \leq E_\kappa$, because taking the optimisation on the third line of~\eqref{E_chi_SM} and restricting to operators $S$ that satisfy $S^\Gamma \geq 0$ yields precisely $E_\kappa$, as can be seen by observing that in those cases $\left\|S^\Gamma \right\|_1 = \Tr S$. Assuming the additivity of $E_\chi$, which we will prove shortly, we immediately deduce that
\bb
E_\kappa \geq E_\kappa^\infty \geq E_\chi\, .
\ee

While this is already interesting in itself, we can do better. In fact, the definition of $E_\chi$ in the second line of~\eqref{E_chi_SM} lends itself naturally to some generalisations. 

\begin{Def}[($\chi$-hierarchy)] \label{chi_hierarchy_Def}
Let $\rho=\rho_{AB}$ be any bipartite quantum state. For $p\in \N$, we define the corresponding \deff{$\boldsymbol{\chi}$-quantity} $E_{\chi,p}$ as 
\begin{align}
E_{\chi, p}(\rho) \coloneqq&\ \log_2 \chi_p(\rho)\, , \label{E_chi_p_SM} \\
\chi_p(\rho) \coloneqq&\ \min \left\{ \Tr S_p :\ -S_i \leq S_{i-1}^\Gamma \leq S_i,\ i=0,...,p,\ S_{-1} = \rho \right\} \label{chi_p_1_SM}\\
=&\ \min\left\{ \Tr S_p :\ -S_0\leq \rho^\Gamma\! \leq S_0,\ -S_1 \leq S_0^\Gamma \leq S_1,\ \ldots\ ,\ -S_p \leq S_{p-1}^\Gamma\!\leq S_p \right\} . \label{chi_p_2_SM}
\end{align}
We will refer to the sequence of functions $E_{\chi,p}$, $p\in \N$, as the \deff{$\boldsymbol{\chi}$-hierarchy}.
\end{Def}

\begin{rem} \label{min_achieved_rem}
The minima in~\eqref{chi_p_1_SM}--\eqref{chi_p_2_SM} are always achieved, essentially because of the compactness of the set of positive semi-definite operators with trace bounded by a fixed number.
\end{rem}

\begin{rem}
For $p>0$, due to Lemma~\ref{variational_trace_norm_lemma} we can equivalently write
\bb
\chi_p(\rho) = \min\left\{ \big\|S_{p-1}^\Gamma\big\| :\ -S_0\leq \rho^\Gamma \leq S_0,\ -S_1 \leq S_0^\Gamma \leq S_1,\ \ldots\ ,\ -S_{p-1} \leq S_{p-2}^\Gamma \leq S_{p-1} \right\}.
\ee
\end{rem}

When $p=0$, clearly we have that
\bb
E_{\chi,0}(\rho) = \log_2 \big\|\rho^\Gamma \big\|_1 = E_N(\rho)
\ee
yields the celebrated logarithmic negativity~\cite{negativity, plenioprl}. For $p=1$, instead,
\bb
E_{\chi,1}(\rho) = E_\chi(\rho)
\ee
reproduces~\eqref{E_chi_SM}. 

While this is quite interesting, we can now observe that the above functions $E_{\chi,p}$ are of a somewhat different nature than the $E_\kappa$ monotone introduced by Wang and Wilde~\cite{Xin-exact-PPT}, whose definition is reported in~\eqref{E_kappa_SM}. We can however try to generalise that construction as well. This is done as follows: for $q\geq 1$, one replaces $p$ with $q$ in~\eqref{chi_p_1_SM}, restricts the minimisation there to operators $S_{q-1}$ satisfying $S_{q-1}^\Gamma \geq 0$, and sets $S_q = S_{q-1}^\Gamma$ accordingly; this gives a `$\kappa$-hierarchy' defined as follows.

\begin{Def}[($\kappa$-hierarchy)] \label{kappa_hierarchy_Def}
Let $\rho=\rho_{AB}$ be any bipartite quantum state. For an arbitrary positive integer $q\in \N^+$, we define the corresponding \deff{$\boldsymbol{\kappa}$-quantity} $E_{\kappa,q}$ as 
\begin{align}
E_{\kappa, q}(\rho) \coloneqq&\ \log_2 \kappa_q(\rho)\, , \label{E_kappa_q_SM} \\
\kappa_q(\rho) \coloneqq&\ \min \left\{ \Tr S_{q-1}:\ -S_i \leq S_{i-1}^\Gamma \leq S_i,\ i=0,...,q-1,\ S_{-1} = \rho,\ S_{q-1}^\Gamma \geq 0 \right\} \label{kappa_q_1_SM} \\
=&\ \min\left\{ \Tr S_{q-1} :\ -S_0\leq \rho^\Gamma\! \leq S_0,\ \ldots\ ,\ -S_{q-1} \leq S_{q-2}^\Gamma\leq S_{q-1},\ S_{q-1}^\Gamma \geq 0 \right\} . \label{kappa_q_2_SM}
\end{align}
We will refer to the sequence of functions $E_{\chi,q}$, $q\in \N^+$, as the \deff{$\boldsymbol{\kappa}$-hierarchy}.
\end{Def}

For reasons analogous to those explained in Remark~\ref{min_achieved_rem}, the minima in~\eqref{kappa_q_1_SM}--\eqref{kappa_q_2_SM} are always achieved. Furthermore, note that
\bb
E_{\kappa,1}(\rho) = E_\kappa(\rho)
\ee
reproduces the function in~\eqref{E_kappa_SM}. We shall now explore the properties of the $\chi$- and $\kappa$-hierarchies in detail.

\subsection{Relation between the \texorpdfstring{$\chi$- and $\kappa$}{chi- and kappa}-hierarchies}

We now begin our investigation of the properties of the $\chi$- and $\kappa$-hierarchies. First, we will study the relation between the two. A very simple yet very useful observation that we record as a preliminary lemma is as follows.

\begin{lemma}[(Increasing trace)] \label{increasing_trace_lemma}
Let the operators $S_0,\ldots, S_p$ be a feasible point for the SDP~\eqref{chi_p_1_SM} that defines $\chi_p(\rho)$, for some bipartite state $\rho = \rho_{AB}$ and some $p\in \N$. Then $S_0,\ldots, S_p\geq 0$ are positive semi-definite, and moreover
\bb
1\leq \Tr S_0 \leq \Tr S_1 \leq \ldots \leq \Tr S_p\, .
\label{increasing_trace_chi}
\ee
Similarly, if $S'_0,\ldots, S'_{q-1}$ constitute a feasible point for the SDP~\eqref{kappa_q_1_SM} that defines $\kappa_q(\rho)$ ($q\in \N^+$), then $S'_0,\ldots, S'_{q-1}\geq 0$, and
\bb
1\leq \Tr S'_0 \leq \Tr S'_1 \leq \ldots \leq \Tr S'_{q-1}\, .
\label{increasing_trace_kappa}
\ee
\end{lemma}

\begin{proof}
From the inequalities $-S_i\leq S_{i-1}^\Gamma \leq S_i$ ($i=0,\ldots, p$, with $S_{-1}=\rho$) it follows immediately that $-S_i\leq S_i$, i.e.\ $S_i\geq 0$. To prove~\eqref{increasing_trace_chi}, it suffices to take the trace of the inequalities $S_0^\Gamma\leq S_1$, $S_1^\Gamma\leq S_2$, \ldots, and $S_{p-1}^\Gamma \leq S_p$. The proof of $S'_0,\ldots, S'_{q-1}\geq 0$ and of~\eqref{increasing_trace_kappa} is entirely analogous.
\end{proof}

\begin{rem} \label{nonzero_rem}
In particular, for all $p\in \N$, $q\in \N^+$, and bipartite states $\rho$, it holds that $\chi_p(\rho)\geq 1$ and $\kappa_q(\rho)\geq 1$. 
\end{rem}

We are now ready to prove the following result.

\begin{prop} \label{relation_hierarchies_prop}
For all non-negative integers $p\in \N$, positive integers $q\in \N^+$, and for all bipartite states $\rho = \rho_{AB}$, it holds that
\bb
E_N = E_{\chi,0} \leq \ldots \leq E_{\chi,p} \leq E_{\chi,p+1} \leq \ldots \leq E_{\kappa,q+1} \leq E_{\kappa,q} \leq \ldots \leq E_\kappa\, ,
\label{relation_hierarchies_SM}
\ee
where we forwent the dependence on $\rho$. In other words, the function $\N\ni p\mapsto E_{\chi,p}(\rho)$ is increasing,\footnote{Not necessarily \emph{strictly} increasing. The meaning of `decreasing' should similarly be intended as `non-increasing'.} while $\N^+\ni q\mapsto E_{\kappa,q}(\rho)$ is decreasing; moreover, $E_{\chi,p}(\rho) \leq E_{\kappa,q}(\rho)$ holds for all $p\in \N$ and $q\in \N^+$.
\end{prop}

\begin{proof}
We prove the claims one by one.
\begin{itemize}
\item \emph{$\N\ni p\mapsto E_{\chi,p}$ is increasing (for fixed $\rho$).} Let $S_0,\ldots, S_{p+1}$ be optimal feasible points for the SDP that defines $\chi_{p+1}$. Then the operators $S_0,\ldots, S_p$ clearly constitute a feasible point for the SDP that defines $\chi_p$, because all the inequalities that have to be obeyed are also present in the SDP that defines $\chi_{p+1}$. Hence, $\chi_p\leq \Tr S_p \leq \Tr S_{p+1} = \chi_{p+1}$, where for the second inequality we leveraged Lemma~\ref{increasing_trace_lemma} (and, more precisely, Eq.~\eqref{increasing_trace_chi}). Taking the logarithm completes the argument.

\item \emph{$\N^+\ni q\mapsto E_{\kappa,q}$ is decreasing (for fixed $\rho$).} The proof is somewhat similar, but reversed. If $S_0,\ldots,S_{q-1}$ represent the optimal feasible point for the SDP~\eqref{kappa_q_1_SM} that defines $\kappa_q$, then we can construct a corresponding feasible point $S_0,\ldots, S_{q-1}, S_q\coloneqq S_{q-1}^\Gamma$ for the SDP that defines $\kappa_{q+1}$. Indeed, the inequalities $-S_i \leq S_{i-1}^\Gamma \leq S_i$ are automatically satisfied for $i=0,\ldots,q-1$ (here, $S_{-1} = \rho$). For $i=q$, we obviously have that $-S_q = -S_{q-1}^\Gamma \leq S_{q-1}^\Gamma = S_{q}$ 
due to the fact that $S_{q-1}^\Gamma \geq 0$. Finally, the last constraint in the SDP for $\kappa_{q+1}$ is $S_q^\Gamma\geq 0$. This is also obeyed, because $S_q^\Gamma = S_{q-1}$, which is positive definite due to Lemma~\ref{increasing_trace_lemma}. Using~\eqref{partial_trace_trace_preserving} we thus obtain that $\kappa_{q+1} \leq \Tr S_q = \Tr S_{q-1}^\Gamma = \Tr S_{q-1} = \kappa_q$. Taking the logarithm completes the proof of the claim.

\item \emph{$E_{\chi,p} \leq E_{\kappa,q}$ for all $p\in \N$ and $q\in \N^+$.} Due to the fact that $\N\ni p\mapsto E_{\chi,p}$ is increasing and 
$\N\ni q \mapsto E_{\kappa,q}$ is decreasing,\footnote{\tcr{The function $\N\ni q \mapsto E_{\kappa,q}$ is also bounded, due to the inequality $E_{\kappa,q}\leq E_\kappa$ together with the fact that $E_\kappa$ is finite for all (finite-dimensional) states.}} the claim is equivalent to the statement that
\bb
\sup_{p\in \N} E_{\chi,p} = \lim_{p\to \infty} E_{\chi,p} \leqt{?} \lim_{q\to\infty} E_{\kappa,q} = \inf_{q\in \N^+} E_{\kappa,q}\, .
\ee
Now, let $S_0,\ldots, S_{q-1}$ be optimal feasible points for the SDP~\eqref{kappa_q_1_SM} that represents $\kappa_q$, for some $q\in \N^+$. By direct inspection of~\eqref{chi_p_1_SM}, we see that these are also feasible points for the SDP that represents $\chi_{q-1}$. Therefore,
\bb
\chi_{q-1} \leq \Tr S_{q-1} = \kappa_q\, .
\ee
Since this holds for an arbitrary $q\in \N^+$, we can now take the limit $q\to\infty$ and subsequently the logarithm of both sides, thus concluding the proof.
\end{itemize}
\end{proof}

To conclude this subsection, we observe that both of our hierarchies collapse into one single quantity in the simple case where the underlying state $\rho$ has zero bi-negativity. This includes, in particular, the case where $\rho$ is pure~\cite[Lemma~5]{Audenaert2002}.

\begin{lemma} \label{normalisation_lemma}
If $\rho = \rho_{AB}$ has zero bi-negativity, in the sense that $\big|\rho^\Gamma\big|^\Gamma\geq 0$, then 
\bb
E_{\chi,p}(\rho) = E_{\kappa,q}(\rho) = E_N(\rho) = \log_2 \big\|\rho^\Gamma\big\|_1
\ee
for all $p\in \N$ and $q\in \N^+$. In particular, for all pure states $\Psi = \ketbra{\Psi}$ with Schmidt decomposition $\ket{\Psi} = \ket{\Psi}_{AB} = \sum_i \sqrt{\lambda_i} \ket{e_i}_A\otimes \ket{f_i}_B$, we have that
\bb
E_{\chi,p}(\Psi) = E_{\kappa,q}(\Psi) = 2 \log_2 \left(\sumno_i \sqrt{\lambda_i} \right) \leq \log_2 d\, ,
\label{hierarchies_pure_states}
\ee
where $d\coloneqq \min\{|A|,|B|\}$, for all $p\in \N$ and $q\in \N^+$, with equality if $\ket{\Psi} = \ket{\Phi_d}$ is the maximally entangled state defined by~\eqref{maximally_entangled_state}.
\end{lemma}

\begin{proof}
Wang and Wilde~\cite[Proposition~8]{Xin-exact-PPT} have already established that $E_\kappa(\rho) = E_N(\rho)$ if $\rho$ has zero bi-negativity. Due to Proposition~\ref{relation_hierarchies_prop}, this implies a complete collapse of both hierarchies. Pure states are known to have zero bi-negativity~\cite[Lemma~5]{Audenaert2002}, hence $E_{\chi,p}(\Psi) = E_{\kappa,q}(\Psi) = E_N(\Psi)$; the explicit expression in terms of the Schmidt coefficients reported in~\eqref{hierarchies_pure_states} is a consequence of the observation that $\big\|\Psi^\Gamma\big\|_1 = \left(\sum_i \sqrt{\lambda_i} \right)^2$ due to~\cite[Proposition~8]{negativity}. In the case of a maximally entangled state of local dimension $d$, we have that $E_N(\Phi_d) = \log_2 d$ because $\lambda_i=1/d$ for all $i=1,\ldots,d$. 
\end{proof}

\subsection{\texorpdfstring{$E_{\chi,{\MakeLowercase p}}$}{} and \texorpdfstring{$E_{\kappa,{\MakeLowercase q}}$}{} as PPT monotones}

We will now prove that our functions $\chi_p$, $\kappa_q$, as well as their logarithmic versions $E_{\chi,p}$ and $E_{\kappa,q}$, represent novel entanglement monotones under PPT operations.

\begin{prop}[(Strong monotonicity)] \label{strong_monotonicity_prop}
For two bipartite quantum systems $AB$ and $A'B'$, let $\big(\NN_x\big)_x$ be a collection of PPT maps with input system $AB$ and output system $A'B'$ such that $\sum_x \NN_x$ is trace preserving. Then for any bipartite quantum state $\rho$ on $AB$ it holds that
\bb
F(\rho) &\geq \sum_{x:\, \Tr\left[\NN_x(\rho)\right]>0}\Tr\left[\NN_x(\rho)\right]\, F\!\left(\frac{\NN_x(\rho)}{\Tr\left[\NN_x(\rho)\right]}\right)\,, \qquad F = \chi_p,\, \kappa_q,\, E_{\chi,p},\, E_{\kappa,q}\, ,\qquad \forall\ p\in \N,\ \forall\ q\in \N^+\, .
\label{strong_monotonicity}
\ee
In particular, for any 
PPT channel $\NN$ it holds that
\bb
E_{\chi,p}(\rho)&\ge E_{\chi,p}\!\left(\NN(\rho)\right)\,,\\
E_{\kappa,q}(\rho)&\ge E_{\kappa,q}\!\left(\NN(\rho)\right)\,.
\ee
\end{prop}

\begin{proof}
The cases $F = E_{\chi,p}$ and $F=E_{\kappa,q}$ of~\eqref{strong_monotonicity} follow by taking the logarithm of the same inequality written for the case $F=\chi_p$ and $F=\kappa_q$, while exploiting the concavity of the function $x\mapsto \log_2(x)$ together with the observation that $\left(\Tr\left[\NN_x(\rho)\right]\right)_x$ is a probability distribution as $\sum_x \NN_x$ is trace preserving.

We thus focus on the cases $F=\chi_p,\kappa_q$, starting from the former. Let $S_{-1},S_0, \ldots, S_p$ be optimisers of the minimisation problem~\eqref{chi_p_1_SM} that defines $\chi_{p}(\rho)$. In particular, 
\bb\label{conditions_feasible}
S_{-1} = \rho,\qquad -S_i\le S_{i-1}^{\Gamma}\le S_i\qquad\forall\ i=0,1,\ldots,p\,.
\ee
Since $\NN_x$ is PPT, the map $\Gamma\circ \NN_x\circ \Gamma$ is (completely) positive. Eq.~\eqref{conditions_feasible} then implies that
\bb \label{condition1_feasible_phi}
-\big[\NN_x(S_{i}^\Gamma)\big]^\Gamma &\le \big[\NN_{x}(S_{i-1})\big]^\Gamma\le \big[\NN_x(S_{i}^\Gamma)\big]^\Gamma\quad\forall\ i=0,1,\ldots,p\,.
\ee
Moreover, Eq.~\eqref{conditions_feasible} and the fact that $\NN_x$ is completely positive imply that
\bb\label{condition2_feasible_phi}
-\NN_x(S_{i})\le \NN_{x}(S_{i-1}^\Gamma)\le \NN_x(S_{i})\qquad\forall\ i=0,1,\ldots,p-1\,.
\ee
Consequently, for any $i=-1,0,\ldots,p$ we can introduce the operators
\bb
S^{(x)}_i \coloneqq
\begin{cases}
    \big[\NN_x(S_i^\Gamma)\big]^\Gamma & \text{if } i \text{ is even}\,, \\[1ex]
    \NN_x(S_i) & \text{if } i\text{ is odd}
\end{cases}
\ee
(here, $-1$ is considered odd), in terms of which we can rewrite~\eqref{condition1_feasible_phi} as
\bb
-S_i^{(x)} \leq \big(S_{i-1}^{(x)}\big)^\Gamma \leq S_i^{(x)} \qquad \text{($i$ even),}
\ee
and~\eqref{condition2_feasible_phi} as 
\bb
-S_i^{(x)} \leq \big(S_{i-1}^{(x)}\big)^\Gamma \leq S_i^{(x)} \qquad \text{($i$ odd).}
\ee
It follows that for any $x$ such that $ \Tr[\NN_x(\rho)]>0$ the operators $S_i^{(x)}\big/\Tr \NN_x(\rho)$, $i=-1,0,\ldots, p$ are valid ansatzes in the minimisation problem~\eqref{chi_p_1_SM} that defines $\chi_p\!\left(\frac{\NN_x(\rho)}{\Tr\left[\NN_x(\rho)\right]}\right)$. Therefore,
\bb
\chi_p(\rho) &= \Tr[S_p] \\
&\eqt{(i)} \sum_x\Tr\big[S_p^{(x)}\big] \\
&\geqt{(ii)} \sum_{x:\, \Tr\left[\NN_x(\rho)\right]>0}\Tr\left[\NN_x(\rho)\right]\,\Tr\left[\frac{S^{(x)}_p}{ \Tr\left[\NN_x(\rho)\right]  }\right] \\
&\geqt{(iii)} \sum_{x:\, \Tr\left[\NN_x(\rho)\right]>0}\Tr\left[\NN_x(\rho)\right]\, \chi_p\!\left(\frac{\NN_x(\rho)}{\Tr\left[\NN_x(\rho)\right]}\right)\,.
\ee
Here, in~(i) we used the fact that the maps $\sum_x \NN_x$ and $\sum_x \Gamma\circ\NN_x\circ\Gamma$ are trace preserving; in~(ii) we observed that $S_p^{(x)}\geq 0$, which follows from~\eqref{condition1_feasible_phi} and~\eqref{condition2_feasible_phi}; finally, in~(iii) we leveraged the fact that the operators $S_i^{(x)}\big/\Tr \NN_x(\rho)$, $i=-1,0,\ldots, p$ are valid ansatzes in the minimisation problem~\eqref{chi_p_1_SM} that defines $\chi_p\!\left(\frac{\NN_x(\rho)}{\Tr\left[\NN_x(\rho)\right]}\right)$, as argued above. This proves~\eqref{strong_monotonicity} for $F=\chi_p$.

The proof of the strong monotonicity property~\eqref{strong_monotonicity} for the case where $F=\kappa_q$ is completely analogous. The only additional observation we need is that $\big(S_{q-1}^{(x)}\big)^\Gamma\geq 0$ if $S_{q-1}^\Gamma \geq 0$. Indeed, if $q$ is odd then $\big(S_{q-1}^{(x)}\big)^\Gamma = \NN_x\big(S_{q-1}^\Gamma\big)$, which is positive because $\NN_x$ is a positive map; if $q$ is even then $\big(S_{q-1}^{(x)}\big)^\Gamma = \NN_x(S_{q-1})^\Gamma = (\Gamma\circ \NN_x\circ \Gamma)\big(S_{q-1}^\Gamma\big)$, which is positive definite because $\Gamma\circ \NN_x\circ \Gamma$ is a positive map. This concludes the proof.
\end{proof}

As an immediate consequence of Lemma~\ref{normalisation_lemma} and Proposition~\ref{strong_monotonicity_prop}, we deduce the following.

\begin{cor} \label{max_values_cor}
Let $\rho = \rho_{AB}$ be an arbitrary quantum state of a bipartite system $AB$ with minimal local dimension $d\coloneqq \min\{|A|,|B|\}$. Then for all $p\in \N$ and $q\in \N^+$ it holds that
\bb
\chi_p(\rho),\, \kappa_q(\rho) \leq d\, ,\qquad E_{\chi,p}(\rho),\, E_{\kappa,q}(\rho)\leq \log_2 d\, .
\label{max_values}
\ee
\end{cor}

\begin{proof}
An application of Nielsen's theorem~\cite{Nielsen-LOCC} guarantees that there exists an LOCC transformation $\Lambda$ such that $\Lambda(\Phi_d) = \rho_{AB}$, where $\Phi_d$ is the maximally entangled state with local dimension $d$.\footnote{Indeed, writing a spectral decomposition $\rho = \sum_i p_i \Psi_i$ of $\rho$, a direct consequence of Nielsen's theorem is that for all $i$ there exists an LOCC $\Lambda_i$ that satisfies $\Lambda_i(\Phi_d) = \Psi_i$. Setting $\Lambda (\cdot) \coloneqq \sum_i p_i \Lambda_i(\cdot)$ proves the claim.} Since LOCC is a subset of PPT, the result then follows by combining Lemma~\ref{normalisation_lemma} and Proposition~\ref{strong_monotonicity_prop}.  Incidentally, an alternative, direct way of constructing an LOCC with the above property is to let the party with the maximum local dimension prepare $\rho_{AB}$ locally and then apply the quantum teleportation protocol~\cite{teleportation} to teleport the smaller system to the other party.
\end{proof}


\subsection{Dual expressions}

In what follows, we derive the SDPs dual to those in~\eqref{chi_p_1_SM} and~\eqref{kappa_q_1_SM}. We refer the reader to~\cite{WATROUS,Skrzypczyk_2023,vandenberghe_1996} for an introduction to the general theory of semi-definite programs. We start with the $\chi$-quantities, whose dual expressions will play an important role in establishing their additivity.

\begin{prop}[(Dual SDP for the $\chi$-quantity)] \label{dual_SDP_chi_prop}
For an arbitrary bipartite state $\rho = \rho_{AB}$ and some $p\in \N$, the associated $\chi$-quantity $\chi_p(\rho)$ defined by~\eqref{chi_p_1_SM} can be equivalently expressed as
\bb \label{chi_p_dual_1_SM}
\begin{array}{llll}
\chi_p(\rho) & = & \mathrm{max.} & \Tr\left[\rho\,\big(V_0^\Gamma-W_0^\Gamma\big)\right] \\[4pt]
& & & \,V_0,\ldots,V_p,W_0,\ldots,W_p\geq 0\, , \\[4pt]
& & \mathrm{s.t.} & V_i+W_i=V_{i+1}^\Gamma-W_{i+1}^\Gamma\quad \forall\ i=0,\ldots,p-1\,, \quad V_p+W_p=\id\,,
\end{array}
\ee
or else, if $p\geq 1$, as
\bb \label{chi_p_dual_2_SM}
\begin{array}{llll}
\chi_p(\rho) & = & \mathrm{max.} & \Tr\left[\rho\,\big(V_0^\Gamma-W_0^\Gamma\big)\right] \\[4pt]
& & & \,V_0,\ldots,V_{p-1},W_0,\ldots,W_{p-1}\geq 0\, , \\[4pt]
& & \mathrm{s.t.} & V_i+W_i=V_{i+1}^\Gamma-W_{i+1}^\Gamma\quad \forall\ i=0,\ldots,p-2\,, \quad -\id\leq \big(V_{p-1}+W_{p-1}\big)^\Gamma\leq \id\,.
\end{array}
\ee
\end{prop}

\begin{rem}
If $p=0$, then the only constraints in~\eqref{chi_p_dual_1_SM} are $V_0,W_0\geq 0$ and $V_0+W_0=\id$. Similarly, if $p=1$ then the only constraints in~\eqref{chi_p_dual_2_SM} are $V_0,W_0\geq 0$ and $-\id\leq \big(V_0+W_0\big)^\Gamma\leq \id$. Note that for $p=1$ the SDP in~\eqref{chi_p_dual_2_SM} reproduces (up to the logarithm) the expression on the first line of~\eqref{E_chi_SM}, as it should. In general,~\eqref{chi_p_dual_2_SM} involves one fewer pair of variables compared to~\eqref{chi_p_dual_1_SM}, as it is deduced by eliminating $V_p$ and $W_p$ from~\eqref{chi_p_dual_1_SM} (see the proof below).
\end{rem}

\begin{proof}[Proof of Proposition~\ref{dual_SDP_chi_prop}]
The $p+1$ (effectively free) variables $S_0,\ldots, S_p$ in~\eqref{chi_p_1_SM} obey $2(p+1)$ semi-definite constraints, namely $S_i\pm S_{i-1}^\Gamma \geq 0$ for $i=0,\ldots, p$, with the convention that $S_{-1}=\rho$. We therefore introduce $p+1$ positive semi-definite dual variables $V_0,\ldots,V_p\geq 0$ for the inequalities with the minus and another $p+1$, namely $W_0,\ldots,W_p\geq 0$, for the inequalities with the plus. The resulting Lagrangian is
\bb
\LL &= \Tr S_p - \sum_{i=0}^p \Tr V_i \big(S_i - S_{i-1}^\Gamma \big) - \sum_{i=0}^p \Tr W_i \big(S_i + S_{i-1}^\Gamma \big)  \\
&= \Tr \rho \big( V_0^\Gamma - W_0^\Gamma\big) - \sum_{i=0}^{p-1} \Tr S_i \big( V_i + W_i - V_{i+1}^\Gamma + W_{i+1}^\Gamma \big) + \Tr S_p \big(\id - V_p - W_p \big)\, ,
\ee
where in the second line we employed~\eqref{partial_trace_HS_preserving} repeatedly. Minimising the above expression over $S_0,\ldots, S_p$ yields $-\infty$ unless $V_i + W_i - V_{i+1}^\Gamma + W_{i+1}^\Gamma = 0$ for all $i=0,\ldots, p-1$, and moreover $V_p+W_p=\id$. This proves that the SDP dual to that in~\eqref{chi_p_1_SM} is the one given by~\eqref{chi_p_dual_1_SM}. 

To see that strong duality holds (and the dual is achieved), it suffices to check that Slater's condition is obeyed, and that both the primal and the dual program are in fact strictly feasible. This simply means that all inequalities can be satisfied strictly. In the primal program~\eqref{chi_p_1_SM}, one can take $S_i = (i+1)\id$ ($i=0,\ldots, p$); in the dual program~\eqref{chi_p_dual_1_SM}, instead, it suffices to set $V_i = 2\cdot 3^{-(p-i+1)} \id$, $W_i = 3^{-(p-i+1)} \id$ ($i=0,\ldots, p$).

Finally, it remains to see that~\eqref{chi_p_dual_2_SM} is equivalent to~\eqref{chi_p_dual_1_SM} if $p\geq 1$. This can be done as follows. In~\eqref{chi_p_dual_1_SM}, the only constraints on $V_p,W_p\geq 0$ are that $V_p - W_p = \big(V_{p-1} - W_{p-1}\big)^\Gamma$ and $V_p+W_p=\id$. It is possible to satisfy these constraints if and only if $-\id\leq \big(V_{p-1} - W_{p-1}\big)^\Gamma\leq \id$. Indeed, this condition is necessary, because on the one hand $\big(V_{p-1} - W_{p-1}\big)^\Gamma \leq V_p \leq V_p+W_p = \id$, and on the other $\big(V_{p-1} - W_{p-1}\big)^\Gamma \geq - W_p \geq -(V_p+W_p) = -\id$. It is also sufficient, since if it is obeyed we can set $V_p = \frac12 \left(\id + \big(V_{p-1} - W_{p-1}\big)^\Gamma \right)$ and $W_p = \frac12 \left(\id - \big(V_{p-1} - W_{p-1}\big)^\Gamma \right)$.
\end{proof}

Although it is not strictly needed in what follows, for completeness we derive also the dual program for the $\kappa$-quantities.

\begin{prop}[(Dual SDP for the $\kappa$-quantity)] \label{dual_SDP_kappa_prop}
For an arbitrary bipartite state $\rho = \rho_{AB}$ and some $q\in \N^+$, the $\kappa$-quantity $\kappa_q(\rho)$ defined by~\eqref{kappa_q_1_SM} can be equivalently expressed as
\bb \label{kappa_q_dual_SM}
\begin{array}{llll}
\kappa_q(\rho) & = & \mathrm{max.} & \Tr\left[\rho\,\big(V_0^\Gamma-W_0^\Gamma\big)\right] \\[4pt]
& & & \,V_0,\ldots,V_{q-1},W_0,\ldots,W_{q-1}\geq 0\, , \\[4pt]
& & \mathrm{s.t.} & V_i+W_i=V_{i+1}^\Gamma-W_{i+1}^\Gamma\quad \forall\ i=0,\ldots,q-2\,, \quad \big(V_{q-1}+W_{q-1}\big)^\Gamma\leq \id\,.
\end{array}
\ee
\end{prop}

\begin{rem}
If $q=1$, then the only constraints in~\eqref{kappa_q_dual_SM} are $V_0,W_0\geq 0$ and $\big(V_0+W_0\big)^\Gamma\leq \id$. Note that for $q=1$ the SDP in~\eqref{kappa_q_dual_SM} reproduces (up to the logarithm) the expression in~\eqref{E_kappa_dual_SM}.
\end{rem}

\begin{proof}[Proof of Proposition~\ref{dual_SDP_kappa_prop}]
The argument is very similar to that employed to prove Proposition~\ref{dual_SDP_chi_prop}. This time around we have $2q+1$ constraints, namely $S_i\pm S_{i-1}^\Gamma \geq 0$ for $i=0,\ldots, q-1$ (with the convention that $S_{-1}=\rho$), as well as $S_{q-1}^\Gamma \geq 0$. We thus introduce $q$ positive semi-definite dual variables $V_0,\ldots,V_{q-1}\geq 0$ for the first inequalities with the minus, another $q$, namely $W_0,\ldots,W_{q-1}\geq 0$, for the first inequalities with the plus, and finally one more variable $Z\geq 0$ for the last inequality. The resulting Lagrangian is
\bb
\LL &= \Tr S_{q-1} - \sum_{i=0}^{q-1} \Tr V_i \big(S_i - S_{i-1}^\Gamma \big) - \sum_{i=0}^{q-1} \Tr W_i \big(S_i + S_{i-1}^\Gamma \big) - \Tr Z S_{q-1}^\Gamma \\
&= \Tr \rho \big( V_0^\Gamma - W_0^\Gamma\big) - \sum_{i=0}^{q-2} \Tr S_i \big( V_i + W_i - V_{i+1}^\Gamma + W_{i+1}^\Gamma \big) + \Tr S_{q-1} \big(\id - Z^\Gamma - V_{q-1} - W_{q-1} \big)\, .
\ee
Minimising the above expression over $S_0,\ldots, S_{q-1}$ gives $-\infty$ unless $V_i + W_i - V_{i+1}^\Gamma + W_{i+1}^\Gamma = 0$ for all $i=0,\ldots, q-2$, and moreover $V_{q-1} - W_{q-1} = \id - Z^\Gamma$. The maximisation over $Z$ can be eliminated by observing that $V_{q-1} - W_{q-1} = \id - Z^\Gamma$ for some $Z\geq 0$ if and only if $\big(V_{q-1} - W_{q-1}\big)^\Gamma \leq \id$. This proves that the SDP in~\eqref{kappa_q_dual_SM} is indeed the dual of that in~\eqref{kappa_q_1_SM}. To see that they yield the same value, i.e.\ that the duality gap is zero, it suffices to check Slater's condition. This can be done with a reasoning entirely analogous to that presented in the proof of Proposition~\ref{dual_SDP_chi_prop}.
\end{proof}

\subsection{Additivity properties} \label{subsec_additivity}

We will now look at the additivity properties of our new monotones $E_{\chi,p}$ and $E_{\kappa,q}$ constructed in Definitions~\ref{chi_hierarchy_Def} and~\ref{kappa_hierarchy_Def}, respectively. We will prove that the former are fully additive (Proposition~\ref{Prop_additivity}) while the latter are only sub-additive (Proposition~\ref{subadditivity_E_kappa_SM_prop}). Before we begin, we need a little preliminary result that is well known in the quantum information literature (see e.g.~\cite[Lemma~12.35]{KHATRI}). We include a self-contained proof for completeness.

\begin{lemma} \label{lemma_tensor_product}
For $i=1,2$, let $A_i$ and $B_i$ be linear operators on the Hilbert space $\HH_i$. Assume that
\bb
-A_i\leq B_i &\leq A_i \qquad \forall\ i=1,2\, .
\ee
Then 
\bb
-A_1\otimes A_2\leq B_1\otimes B_2\leq A_1\otimes A_2\,.
\label{tensor_product_sandwhich}
\ee
In particular, if $-A\leq B\leq A$ then for all $n\in\N$ it holds that $-A^{\otimes n}\leq B^{\otimes n}\leq A^{\otimes n}$.
\end{lemma}

\begin{proof}
By adding the two operator inequalities
\bb
0\leq (A_1+B_1)\otimes (A_2+B_2) &= A_1\otimes A_2 + B_1\otimes B_2 + A_1\otimes B_2 + B_1\otimes A_2\,,\\
0\leq (A_1-B_1)\otimes (A_2-B_2) &= A_1\otimes A_2 + B_1\otimes B_2 - A_1\otimes B_2 - B_1\otimes A_2\,,\\
\ee
we infer that $B_1\otimes B_2\geq - A_1\otimes A_2$. Analogously, by adding the two operator inequalities
\bb
0\leq (A_1-B_1)\otimes (A_2+B_2) &= A_1\otimes A_2 - B_1\otimes B_2 + A_1\otimes B_2 - B_1\otimes A_2\,,\\
0\leq (A_1+B_1)\otimes (A_2-B_2) &= A_1\otimes A_2 - B_1\otimes B_2 - A_1\otimes B_2 + B_1\otimes A_2\,,\\
\ee
one obtains that $B_1\otimes B_2\leq A_1\otimes A_2$. The last claim is applying~\eqref{tensor_product_sandwhich} iteratively $n-1$ times, indexed $j=1,\ldots,n-1$, with $A_1 = A^{\otimes j}$, $B_1 = B^{\otimes j}$, $A_2=A$, and $B_2=B$.
\end{proof}

We are now ready to state and prove the additivity of $E_{\chi,p}$.

\begin{prop}[(Additivity of $E_{\chi,p}$)]\label{Prop_additivity}
For all pairs of bipartite states $\rho_{AB},\,\omega_{A'B'}$ and for all $p\in \N$, we have that
\begin{align}
\chi_p\left(\rho_{AB}\otimes \omega_{A'B'}\right) &= \chi_p(\rho_{AB})\, \chi_p(\omega_{A'B'})\, , \label{chi_multiplicativity_SM} \\
E_{\chi, p}\left(\rho_{AB}\otimes \omega_{A'B'}\right) &= E_{\chi, p}\left(\rho_{AB}\right) + E_{\chi, p}\left(\omega_{A'B'}\right) .
\label{E_chi_additivity_SM}
\end{align}
\end{prop}

\begin{proof}
Clearly, Eq.~\eqref{E_chi_additivity_SM} follows from~\eqref{chi_multiplicativity_SM} upon taking the logarithm. We therefore focus on~\eqref{chi_multiplicativity_SM}. 
To prove it, we will show that $\chi_p$ is both sub-multiplicative and super-multiplicative. Let us start by verifying the former property, i.e.\ the inequality
\bb
\chi_p(\rho\otimes \omega) \leqt{?} \chi_p(\rho)\,\chi_p(\omega)\, , 
\ee
where we omitted the system labels for simplicity. Let $S_0,\ldots, S_p$ be optimisers of the minimisation problem~\eqref{chi_p_1_SM} that defines $\chi_p(\rho)$. Analogously, let $S'_0,\ldots, S'_p$ be optimisers of the minimisation problem that defines $\chi_p(\omega)$. By exploiting Lemma~\ref{lemma_tensor_product}, one can easily check that $S_0\otimes S'_0 ,\ldots, S_p\otimes S'_p$ are valid ansatzes in the minimisation problem that defines $\chi_p(\rho\otimes \omega)$. Consequently,
\bb
\chi_p(\rho\otimes \omega) \leq \Tr\big[S_p\otimes S'_p\big] = \chi_p(\rho)\,\chi_p(\omega)\,.
\ee

Now, let us prove the super-multiplicativity of $\chi_p$. We thus need to verify that 
\bb
\chi_p(\rho\otimes \omega) \geqt{?} \chi_p(\rho)\,\chi_p(\omega)\, .
\ee
To this end, we need to employ the dual SDP for $\chi_p$ that we derived in Proposition~\ref{dual_SDP_chi_prop}. Let $V_0,\ldots,V_p,W_0,\ldots,W_p \geq 0$ be optimisers of the SDP in~\eqref{chi_p_dual_1_SM} for $\chi_p(\rho)$. Analogously, let $V'_0,\ldots,V'_p,W'_0,\ldots,W'_p \geq 0$ be optimisers of the SDP as in~\eqref{chi_p_dual_1_SM} but for the state $\omega$. For any $i=0,\ldots,p$, define
\bb
V''_i &\coloneqq V_i \otimes V'_i + W_i\otimes W'_i\,, \\
W''_i &\coloneqq V_i \otimes W'_i + W_i \otimes V'_i\,.
\ee
Note that
\bb
V''_i + W''_i &= \left(V_i + W_i \right)\otimes\left(V'_i + W'_i \right) ,\\
V''_i - W''_i &= \left(V_i - W_i \right)\otimes\left(V'_i - W'_i \right) ,
\ee
one can easily check that $V''_0,\ldots,V''_p, W''_0,\ldots,W''_p$ are valid ansatzes for the optimisation as in~\eqref{chi_p_dual_1_SM} but for the state $\rho\otimes \omega$. Indeed, for $i=0,\ldots, p$ we have that $V''_i,W''_i\geq 0$; moreover, if $i\leq p-1$ it also holds that
\bb
V''_i + W''_i &= \left(V_i + W_i \right)\otimes\left(V'_i + W'_i \right) \\
&= \left(V_{i+1} - W_{i+1} \right)^\Gamma\otimes\left(V'_{i+1} - W'_{i+1} \right)^\Gamma \\
&= \left( \left(V_{i+1} - W_{i+1} \right)\otimes\left(V'_{i+1} - W'_{i+1} \right) \right)^\Gamma \\
&= \left(V''_{i+1} - W''_{i+1}\right)^\Gamma ;
\ee
finally, for $i=p$ we see that 
\bb
V''_p + W''_p = \left(V_p + W_p \right)\otimes\big(V'_p + W'_p \big) = \id\otimes \id = \id\, .
\ee
Consequently, it holds that
\bb
\chi_p(\rho\otimes \omega) &\geq \Tr\left[\left(\rho\otimes\omega\right)\left(V''_0 - W''_0 \right)^\Gamma\right]\\
&= \Tr\left[\rho\,\big(V_0 - W_0 \big)^\Gamma\right]\, \Tr\left[\omega\,\big( V'_0 - W'_0 \big)^\Gamma\right]\\
&= \chi_p(\rho)\,\chi_p(\omega)\,.
\ee
This concludes the proof.
\end{proof}

Unlike the $\chi$-quantities, the corresponding $\kappa$-quantities are only sub-additive but in general not fully additive (as we have seen in Section~\ref{sec_additivity_violations}). Although we will not make use of this property in what follows, we state and prove it for completeness.

\begin{prop}[(Sub-additivity of $E_{\kappa,q}$)] \label{subadditivity_E_kappa_SM_prop}
For all pairs of bipartite states $\rho_{AB},\,\omega_{A'B'}$ and for all $q\in \N^+$, we have that
\begin{align}
\kappa_q\left(\rho_{AB}\otimes \omega_{A'B'}\right) &\leq \kappa_q(\rho_{AB})\, \kappa_q(\omega_{A'B'})\, , \label{submultiplicativity_kappa_SM} \\
E_{\kappa, q}\left(\rho_{AB}\otimes \omega_{A'B'}\right) &\leq E_{\kappa, q}\left(\rho_{AB}\right) + E_{\kappa, q}\left(\omega_{A'B'}\right) .
\label{subadditivity_E_kappa_SM}
\end{align}
\end{prop}

\begin{proof}
The argument is totally analogous to that in the first part of the proof of the above Proposition~\ref{Prop_additivity}. Let $S_0,\ldots, S_{q-1}$ and $S'_0,\ldots, S'_{q-1}$, respectively, be optimisers of the SDPs that define $\kappa_q(\rho)$ and $\kappa_q(\omega)$ as in~\eqref{kappa_q_1_SM}. Due once more to Lemma~\ref{lemma_tensor_product}, one sees that $S_0\otimes S'_0, \ldots, S_{q-1}\otimes S'_{q-1}$ constitute a feasible point for the SDP problem that defines $\kappa_q(\rho\otimes \omega)$. Consequently,
\bb
\kappa_q(\rho\otimes \omega) \leq \Tr\left[S_{q-1}\otimes S'_{q-1}\right] = \kappa_q(\rho)\,\kappa_q(\omega)\,.
\ee
This concludes the proof.
\end{proof}

\subsection{Continuity} \label{subsec_continuity}

To close off this section, we explore the continuity properties of our new monotones. We need a couple of simple lemmas first. To simplify the exposition, it is useful to note that the functions $\chi_p$ and $\kappa_q$ ($p\in \N$, $q\in \N^+$) 
given by~\eqref{chi_p_1_SM} and~\eqref{kappa_q_1_SM} are perfectly well defined and non-negative not only for states, but also for arbitrary Hermitian operators. In fact, looking at the SDPs in~\eqref{chi_p_1_SM} and~\eqref{kappa_q_1_SM}, one sees that $\chi_p(X)\geq 0$ and $\kappa_q(X)\geq 0$ must hold for all Hermitian $X$, because in both cases any feasible point must satisfy $S_i\geq 0$ for all $i$.

It is 
easy to verify that
\bb
F(X+Y) \leq F(X) + F(Y)\, ,\qquad F = \chi_p,\, \kappa_q,\quad p\in \N,\ q\in \N^+ ,
\label{subadditivity_F}
\ee
and moreover 
\bb
F(\lambda X) = |\lambda|\,F(X)\, ,\qquad F = \chi_p,\, \kappa_q,\quad p\in \N,\ q\in \N^+ ,\quad \lambda\in \R\,.
\label{invariance_F}
\ee
Here, the absolute value in~\eqref{invariance_F} is a consequence of the fact that 
$X$ appears with both signs in the definition of both $\chi_p(X)$ and $\kappa_q(X)$. We can now extend Lemma~\ref{normalisation_lemma} to the non-positive case.

\begin{lemma} \label{extended_normalisation_lemma}
For all Hermitian operators $X = X_{AB}$ on a bipartite quantum system of local dimension $d\coloneqq \min\{|A|,|B|\}$, it holds that
\bb
F(X) \leq d\,\|X\|_1\, ,\qquad F = \chi_p,\, \kappa_q,\quad p\in \N,\ q\in \N^+\, .
\ee
\end{lemma}

\begin{proof}
Due to the 
properties of $F$, for all Hermitian operators $X$ we have that 
\bb
F(X)\ \eqt{(i)}\ F\left(\sumno_i x_i \Psi_i\right)\ \leqt{(ii)}\ \sum_i F(x_i \Psi_i)\ \eqt{(iii)}\ \sum_i |x_i|\, F(\Psi_i)\ \leqt{(iv)}\ d \sum_i |x_i| = d\, \|X\|_1\, .
\ee
Here, in~(i) we introduced a spectral decomposition $X=\sum_i x_i \Psi_i$ of $X$, (ii)~follows from~\eqref{subadditivity_F}, (iii) descends from~\eqref{invariance_F}, and finally in~(iv) we leveraged Lemma~\ref{normalisation_lemma}, and in particular~\eqref{hierarchies_pure_states}.
\end{proof}

We are now ready to state the continuity properties of our monotones.

\begin{prop}[(Continuity)] \label{continuity_properties_prop}
Let $\rho,\rho'$ be two quantum states on a bipartite system $AB$ with minimal local dimension $d\coloneqq \min\{|A|,|B|\}$. Setting $\e\coloneqq \frac12 \left\|\rho-\rho'\right\|_1$, for all $p\in \N$ and $q\in \N^+$ we have that
\bb
\left|F(\rho) - F(\rho') \right| \leq 2d\e\, ,\qquad F = \chi_p,\, \kappa_q\, .
\label{continuity_F}
\ee
Consequently,
\bb
\left|f(\rho) - f(\rho') \right| \leq \log_2\left(1+2d\e\right) \leq 2\,(\log_2 e)\, d\e ,\qquad f = E_{\chi,p},\, E_{\kappa,q}\, .
\label{continuity_f}
\ee
\end{prop}

\begin{proof}
Let $X\coloneqq \rho' - \rho$. Since $\|X\|_1 = 2\e$, we can write
\bb
F(\rho') = F(\rho + X) \leq F(\rho) + F(X) \leq F(\rho) + 2d \e\, ,
\ee
where we have used~\eqref{subadditivity_F} and Lemma~\ref{extended_normalisation_lemma}. Applying this reasoning with $\rho$ and $\rho'$ exchanged yields also the reverse inequality and thus completes the proof of~\eqref{continuity_F}. As for~\eqref{continuity_f}, it suffices to note that
\bb\label{eq_proof_continuity}
F(\rho') \leq F(\rho) +2d\e \leq F(\rho) \left(1+2d\e\right) ,
\ee
where in the last line we used that $F(\rho)\geq 1$ by Remark~\ref{nonzero_rem}. The claim follows by taking logarithms in~\eqref{eq_proof_continuity}, and 
subsequently exchanging the roles of 
$\rho$ and $\rho'$.

\end{proof}

\begin{rem}
The type of continuity stated in Proposition~\ref{continuity_properties_prop} is substantially weaker than \emph{asymptotic continuity}, a key property of entanglement monotones~\cite{Horodecki2001,Donald2002}. The main point is that the Lipschitz constant appearing on the right hand side of~\eqref{continuity_F}--\eqref{continuity_f} is proportional to $d$ instead of $\log_2 d$, as would be needed to have asymptotic continuity. This is not surprising, because the logarithmic negativity itself is known to be \emph{not} asymptotically continuous~\cite{negativity}, and our monotones can be thought of generalisations of the negativity. It is worth remarking, however, that even the weaker form of continuity established in Proposition~\ref{continuity_properties_prop} can be extremely useful. The same property for the negativity, for example, which can be re-derived by setting $F=\chi_0$ in~\eqref{continuity_F}, is key to proving that it upper bounds the distillable entanglement~\cite[Section~IV]{negativity}. This observation also underpins the whole approach of~\cite{irreversibility}.
\end{rem}

\section{Main results}

\subsection{Enter the regularisation}


In the previous section we have investigated two families of new PPT monotones, the $\chi$- and the $\kappa$-hierarchies. In Proposition~\ref{relation_hierarchies_prop} we have established a precise hierarchical relation between the two hierarchies, with the increasing functions $E_{\chi,p}(\rho)$ ($p\in \N$) lying below the decreasing functions $E_{\kappa,q}(\rho)$ ($q\in \N^+$). However, our ultimate goal is to understand the zero-error PPT entanglement cost of an arbitrary state $\rho$, which, by the results of Wang and Wilde (Eq.~\eqref{cost_E_kappa_regularised}), coincides with the regularised $E_\kappa$, in formula $\ecost(\rho) = E_\kappa^\infty(\rho)$. The purpose of this subsection, therefore, is to understand how $E_\kappa^\infty$ fits in the hierarchy delineated by Proposition~\ref{relation_hierarchies_prop}. A first insight can be deduced by leveraging the results we have obtained in the previous section.

\begin{lemma} \label{cost_lower_bounded_E_chi_p_lemma}
Let $\rho=\rho_{AB}$ be an arbitrary bipartite state. Then for all $p\in \N$ it holds that
\bb
\ecost(\rho) = E_\kappa^\infty(\rho) \geq E_{\chi,p}(\rho)\, .
\label{cost_lower_bounded_E_chi_p}
\ee
In particular,
\bb
E_\kappa^\infty(\rho) \geq \lim_{p\to\infty} E_{\chi,p}(\rho) = \sup_{p\in \N} E_{\chi,p}(\rho)\, .
\label{cost_lower_bounded_limit_E_chi_p}
\ee
\end{lemma}

\begin{proof}
For all positive integers $n\in \N^+$ we can write
\bb
\frac1n\, E_\kappa\big(\rho^{\otimes n}\big) \geqt{(i)} \frac1n\, E_{\chi,p}\big(\rho^{\otimes n}\big) \eqt{(ii)} E_{\chi,p}(\rho)\, ,
\ee
where (i)~holds due to Proposition~\ref{relation_hierarchies_prop}, and in~(ii) we exploited the additivity of $E_{\chi,p}$ (Proposition~\ref{Prop_additivity}). Taking the limit $n\to\infty$ and leveraging~\eqref{cost_E_kappa_regularised} concludes the proof of~\eqref{cost_lower_bounded_E_chi_p}. Upon taking the limit $p\to\infty$, we obtain also~\eqref{cost_lower_bounded_limit_E_chi_p}.
\end{proof}

We have thus established that the function $E_\kappa^\infty$ lies above the entire $\chi$-hierarchy. We shall now investigate its relationship with the $\kappa$-hierarchy. The main result of this subsection is the following.

\begin{thm} \label{hierarchies_chi_kappa_infty_thm}
Let $q\in \N^+$ be a positive integer. For any bipartite state $\rho = \rho_{AB}$, it holds that
\bb
\ecost(\rho) = E^\infty_\kappa(\rho) \leq E_{\kappa, q}(\rho)\, .
\label{cost_upper_bounded_E_kappa_q}
\ee
In particular, for every $\rho$ we have that
\bb
E_N = E_{\chi,0} \leq E_{\chi,1} \leq \ldots \leq \lim_{p\to\infty} E_{\chi,p} \leq \ecost = E^\infty_\kappa \leq \lim_{q\to\infty} E_{\kappa, q} \leq \ldots \leq E_{\kappa,2} \leq E_{\kappa,1} = E_\kappa\, .
\label{hierarchies_chi_kappa_infty}
\ee
\end{thm}

\begin{proof}
Clearly, it suffices to prove~\eqref{cost_upper_bounded_E_kappa_q}, as~\eqref{hierarchies_chi_kappa_infty} would then follow by combining that with the results of Lemma~\ref{cost_lower_bounded_E_chi_p_lemma} and Proposition~\ref{relation_hierarchies_prop}. In light of these considerations, we thus set out to prove~\eqref{cost_upper_bounded_E_kappa_q}. 

Let $S_{-1},S_0,\ldots,S_{q-1}$ be optimisers for the SDP in~\eqref{kappa_q_1_SM} for $\kappa_q(\rho)$. Therefore,
\bb \label{feasible_points}
-S_i \leq S_{i-1}^\Gamma \leq S_i\qquad \forall\ i=0,\ldots, q-1\, ,\qquad S_{-1} =\rho\, ,\qquad S_{q-1}^\Gamma \geq 0\,.
\ee
For an arbitrary positive integer $n\in\N^+$, Lemma~\ref{lemma_tensor_product} implies that
\bb \label{ineq_s_i_otimes}
-S_i^{\otimes n} \leq \big(S_{i-1}^{\otimes n}\big)^\Gamma \leq S_i^{\otimes n} \qquad \forall\ i=0,\ldots, q-1\,, \qquad 
\big(S_{q-1}^{\Gamma}\big)^{\otimes n} \geq 0
\ee
Let us define
\bb \label{def_Sn_new}
\tilde{S}_n\coloneqq \sum_{i=0}^{q-1}\left(S_i^{\otimes n}+\big(S_i^\Gamma\big)^{\otimes n}\right) .
\ee
We will now show that $\tilde{S}_n$ is in fact a feasible point of the minimisation problem that defines $\kappa_1\big(\rho^{\otimes n}\big) = 2^{E_\kappa(\rho^{\otimes n})}$. To this end, we need to verify that  
\bb
-\tilde{S}_n\leqt{?} \big(\rho^{\otimes n}\big)^\Gamma \leqt{?} \tilde{S}_n\,,\qquad \tilde{S}_n^{\,\Gamma}\geqt{?}0\,.
\label{inequalities_to_be_proved_Sn_new}
\ee
To prove the first two inequalities, note that 
\bb
\tilde{S}_n \pm \big(\rho^{\otimes n}\big)^\Gamma &= S_0^{\otimes n} + \sum_{i=0}^{q-2}\left(S_{i+1}^{\otimes n}+\big(S_i^\Gamma\big)^{\otimes n}\right)+ \big(S_{q-1}^\Gamma\big)^{\otimes n} \pm \big(\rho^{\otimes n}\big)^\Gamma \\
&\geq S_0^{\otimes n}+\big(S_{q-1}^\Gamma\big)^{\otimes n} \pm \big(\rho^{\otimes n}\big)^\Gamma \\
&\geq S_0^{\otimes n} \pm \big(\rho^{\otimes n}\big)^\Gamma \\
&\geq 0\, ,
\ee
where all three inequalities come from~\eqref{ineq_s_i_otimes}. 
Consequently, we have proved that $-\tilde{S}_n\leq \big(\rho^{\otimes n}\big)^\Gamma \leq \tilde{S}_n$. Note that 
these two inequalities together imply that $\tilde{S}_n \geq 0$ (because $\tilde{S}_n \geq -\tilde{S}_n$). In addition, by exploiting the definition of $\tilde{S}_n$ in~\eqref{def_Sn_new}, we also have that $\tilde{S}_n^\Gamma = \tilde{S}_n$, and hence $\tilde{S}_n^{\,\Gamma}\geq 0$.  We have thus proved that $\tilde{S}_n$ is a feasible point of the minimisation problem that defines (up to a logarithm) $E_\kappa\big(\rho^{\otimes n}\big)$. Hence, it follows that
\bb \label{proof_step}
    E_\kappa\big(\rho^{\otimes n}\big) &\leq \log_2 \Tr \tilde{S}_n\\
    &= \log_2 \Tr\left[ \sum_{i=0}^{q-1}\left(S_i^{\otimes n}+\big(S_i^\Gamma\big)^{\otimes n}\right)\right]\\
    &\eqt{(i)} \log_2\left( 2\sum_{i=0}^{q-1}(\Tr S_i)^n\right)\\
    &\leqt{(ii)} \log_2 \left(2q \left(\Tr S_{q-1}\right)^n\right)\\
    &= \log_2(2q)+ n\log_2(\Tr S_{q-1}) \\
    &\eqt{(iii)} \log_2(2q)+ n\, E_{\kappa,q}(\rho)\, .
\ee
Here, in~(i) we exploited~\eqref{partial_trace_trace_preserving}, (ii)~follows from the inequalities $\Tr S_0 \leq \Tr S_1 \leq \ldots \leq \Tr S_{q-1}$ (Lemma~\ref{increasing_trace_lemma}), and finally~(iii) holds by the optimality of $S_{-1},S_0,\ldots,S_{q-1}$. As a consequence of~\eqref{proof_step}, we have that for all fixed $q\in \N^+$
\bb
    E^\infty_\kappa(\rho)=\lim\limits_{n\rightarrow\infty} \frac1n\, E_\kappa\big(\rho^{\otimes n}\big) \leq E_{\kappa,q}(\rho)\, ,
\ee
which establishes~\eqref{cost_upper_bounded_E_kappa_q} and thus concludes the proof.
\end{proof}

\begin{rem}
Theorem~\ref{hierarchies_chi_kappa_infty_thm} also shows that
\bb
E_{\kappa,q}^\infty(\rho) \coloneqq \lim_{n\to\infty} \frac1n\, E_{\kappa,q}\big(\rho^{\otimes n}\big) = E_{\kappa}^\infty(\rho) = \ecost(\rho)\qquad \forall\ \rho,\quad \forall\ q\in \N^+\, .
\ee
\end{rem}

\subsection{A key technical result: proof of Proposition~\ref{convergence_primitive_prop}}

This subsection is devoted to the detailed proof of the crucial Proposition~\ref{convergence_primitive_prop}, which links the $\chi$- and the $\kappa$-hierarchies in a profound way. This connection will be the keystone on which all of our main results are build.

\begin{manualprop}{\ref{convergence_primitive_prop}} 
For all bipartite states $\rho=\rho_{AB}$ on a system of minimal local dimension $d \coloneqq \min\{|A|,|B|\} \geq 2$, and all positive integers $p\in \N^+$, it holds that
\bb
\kappa_p(\rho) \leq \chi_p(\rho) + \left(\frac{d}{2}-1\right) \left(\chi_p(\rho) - \chi_{p-1}(\rho) \right)
\label{convergence_primitive_SM}
\ee
In particular,
\bb
E_{\chi,p}(\rho) \geq \frac{2}{d} E_{\kappa,p}(\rho) + \left(1-\frac{2}{d}\right) E_{\chi,p-1}(\rho)\, .
\label{convergence_logs_SM}
\ee
\end{manualprop}

To obtain the above result we first need a key technical lemma, proved below. Before stating it, let us fix some terminology. Given an arbitrary positive semi-definite bipartite operator $T=T_{AB}$, we set
\bb
d_{\max}\left(T\|\PPT\right) \coloneqq \min\left\{ \Tr L:\ T\leq L\, ,\ \ L\in \PPT\right\} = \min\left\{ \Tr L:\ T\leq L,\ L\geq 0,\ L^\Gamma \geq 0 \right\} .
\label{d_max_PPT}
\ee
It is well known that for all positive semi-definite $T$ it holds that~\cite{VidalTarrach}
\bb
d_{\max}\left(T\|\PPT\right) \leq d \Tr T\, .
\label{robustness_PPT_bound}
\ee
One way to see this is, similarly to Corollary~\ref{max_values_cor}, by first noticing that for the maximally entangled state one has $d_{\max}(\Phi_d\|\PPT) = d$ and then using the fact that any state $T/\Tr T$ can be obtained from $\Phi_d$ by an LOCC transformation, which can never increase $d_{\max}(\cdot\|\PPT)$. For completeness, let us give here a more self-contained proof of Eq.~\eqref{robustness_PPT_bound}.

Up to taking positive linear combinations, it suffices to prove~\eqref{robustness_PPT_bound} when $T=\ketbra{\psi}$ is the rank-one projector onto the pure state $\ket{\psi}=\ket{\psi}_{AB}$ with Schmidt decomposition $\ket{\psi}_{AB} = \sum_{i=1}^d \sqrt{\lambda_i} \ket{e_i}_A\ket{f_i}_B$. Setting $L \coloneqq 
\ketbra{\psi}_{AB} + \sum_{i\neq j} \sqrt{\lambda_i \lambda_j} \ketbraa{e_i}{e_i}_A \otimes \ketbraa{f_j}{f_j}_B$ one sees that on the one hand $L\geq T = \ketbra{\psi}$, while on the other 
\bb
L^\Gamma \simeq \bigoplus_i \lambda_i \oplus \bigoplus_{i<j} \sqrt{\lambda_i \lambda_j} \begin{pmatrix}1&1\\1&1\end{pmatrix}
\geq 0\, .
\ee
Then, using Cauchy--Schwarz and noting that $\sum_i\lambda_i = 1$ one finds that
\bb
\Tr L = 1 + \sum_{i\neq j} \sqrt{\lambda_i\lambda_j} = \sum_i \lambda_i + \sum_{i\neq j} \sqrt{\lambda_i\lambda_j} = \left(\sumno_i \sqrt{\lambda_i}\right)^2 \leq d\, ,
\ee
which concludes the proof of~\eqref{robustness_PPT_bound}.

The key lemma that is needed to prove Proposition~\ref{convergence_primitive_prop} is the following.

\begin{lemma} \label{chi_to_kappa_lemma}
For $p\geq 1$, let $S_{-1},S_0,\ldots, S_{p-1}, S_p$ be optimisers for the SDP~\eqref{chi_p_1_SM} that defines $\chi_p(\rho)$. 
Then there exists a PPT operator $M=M_{AB}\in \PPT$ such that $S_{-1},S_0,\ldots, S_{p-2}, S_{p-1} + M$ 
are valid ansatzes for the SDP that defines $\kappa_p(\rho)$ as in~\eqref{kappa_q_1_SM}, and moreover $\Tr M = d_{\max}\big( \big(S_{p-1}^\Gamma\big)_- \,\big\|\, \PPT\big)$. 
\end{lemma}

\begin{proof}
Since $-S_i\leq S_{i-1}^\Gamma \leq S_i$ holds for all $i=0,\ldots, p-1$ by assumption, and moreover $M\geq 0$, the only inequality that is left to verify is $\left(S_{p-1} + M \right)^\Gamma\geq 0$. By definition of $d_{\max}(\cdot\|\PPT)$, 
there exists a PPT operator $L=L_{AB}$ such that $\big(S_{p-1}^\Gamma\big)_-\leq L$ and $\Tr L = d_{\max}\big( \big(S_{p-1}^\Gamma\big)_- \,\big\|\, \PPT\big)$. We can then define $M\coloneqq L^\Gamma$, so that
\bb
\left(S_{p-1} + M \right)^\Gamma = S_{p-1}^\Gamma + L = \big(S_{p-1}^\Gamma\big)_+ - \big(S_{p-1}^\Gamma\big)_- + L \geq \big(S_{p-1}^\Gamma\big)_+ \geq 0\, .
\ee
Since $\Tr M = \Tr L = d_{\max}\big( \big(S_{p-1}^\Gamma\big)_- \,\big\|\, \PPT\big)$, this concludes the proof of the lemma.
\end{proof}




We are now ready to give the complete proof of Proposition~\ref{convergence_primitive_prop}.

\begin{proof}[Proof of Proposition~\ref{convergence_primitive_prop}]
Let $S_{-1},S_0,\ldots, S_{p-1}, S_p$ be optimisers for the SDP~\eqref{chi_p_1_SM} that defines $\chi_p(\rho)$. Then clearly $S_{-1},S_0,\ldots, S_{p-1}$ are valid ansatzes for the SDP that defines $\chi_{p-1}(\rho)$. Therefore,
\bb
\Tr S_{p-1} \geq \chi_{p-1}(\rho) = \chi_p(\rho) - \Delta_p = \Tr S_p - \Delta_p\, ,
\ee
where we defined $\Delta_p \coloneqq \chi_p(\rho) - \chi_{p-1}(\rho)$. We deduce that
\bb
0\leq \delta_p \coloneqq \Tr \left[ S_p - S_{p-1} \right] \leq \Delta_p\, .
\label{delta_vs_Delta}
\ee
The first inequality is a consequence of Lemma~\ref{relation_hierarchies_prop} --- basically, of the fact that $\delta_p = \Tr \big[ S_p - S_{p-1}^\Gamma \big]$, and the operator inside the trace is positive semi-definite. Now, by Lemma~\ref{chi_to_kappa_lemma} there exists a PPT operator $M$ such that $S_{-1},S_0,\ldots, S_{p-2}, S_{p-1} + M$ are legitimate ansatzes for the SDP that defines $\kappa_p(\rho)$ as in~\eqref{kappa_q_1_SM}, and moreover
\bb
\Tr M &= d_{\max}\big( \big(S_{p-1}^\Gamma\big)_- \,\big\|\, \PPT\big) \\
&\leqt{(i)} d \Tr\!\big[ \big(S_{p-1}^\Gamma\big)_- \big] \\
&= d \, \frac{\big\|S_{p-1}^\Gamma\big\|_1 - \Tr S_{p-1}}{2} \\
&\leqt{(ii)} d\, \frac{\Tr S_p - \Tr S_{p-1}}{2} \\
&= \frac{d \delta_p}{2}\, ,
\label{trace_M_estimate}
\ee
where the bound in~(i) follows from~\eqref{robustness_PPT_bound}, and that in~(ii) comes from Lemma~\ref{variational_trace_norm_lemma} applied to the operator inequalities $-S_p \leq S_{p-1}^\Gamma\leq S_p$. (The latter is, in fact, an equality if the operators $S_i$ are optimisers for $\chi_p(\rho)$; however, we shall not make use of this observation.)

Putting all together, we obtain that
\bb
\kappa_p(\rho) &\leq \Tr \left[ S_{p-1}+ M \right] \\
&= \Tr S_p - \delta_p + \Tr M \\
&\eqt{(iii)} \chi_p(\rho) - \delta_p + \Tr M \\
&\leqt{(iv)} \chi_p(\rho) + \left(\frac{d}{2}-1\right) \delta_p \\
&\leqt{(v)} \chi_p(\rho) + \left(\frac{d}{2}-1\right) \Delta_p \\
&= \chi_p(\rho) + \left(\frac{d}{2}-1\right) \left(\chi_p(\rho) - \chi_{p-1}(\rho)\right) ,
\ee
which reproduces~\eqref{convergence_primitive_SM}. Here, (iii)~follows by the optimality of $S_{-1},S_0,\ldots, S_p$ for $\chi_p(\rho)$, in~(iv) we employed~\eqref{trace_M_estimate}, and finally~(v)~follows from~\eqref{delta_vs_Delta}.

Rearranging, taking the logarithms, and using the concavity of the $\log_2$ function, one obtains also~\eqref{convergence_logs_SM}. 
\end{proof}

\subsection{What have we achieved so far?}
\phantomsection

\subsubsection{Convergence}

Proposition~\ref{convergence_primitive_prop} is the core technical finding on which all of our main results rest. Now that we have proved it, most of the remaining proofs are comparatively 
straightforward. To show the importance of Proposition~\ref{convergence_primitive_prop}, it is instructive to pause for a second and explore some of its immediate consequences. In this subsection we do precisely that. 

We first show that Proposition~\ref{convergence_primitive_prop} already implies~\eqref{convergence_simplified}, i.e.\ that the limits of the $\chi$- and $\kappa$-hierarchies coincide and yield the true zero-error PPT entanglement cost. In fact, taking the limit $p\to\infty$ on both sides of~\eqref{convergence_primitive_SM} and using the fact that both $\lim_{p\to\infty} \chi_p(\rho)$ and $\lim_{p\to\infty} \kappa_p(\rho)$ exist by monotonicity (Proposition~\ref{relation_hierarchies_prop}) shows that
\bb
\lim_{p\to\infty} \kappa_p(\rho) \leq \lim_{p\to\infty} \chi_p(\rho)\, .
\ee
Combined with the reverse inequality, which is an obvious consequence of~\eqref{relation_hierarchies_SM}, one deduces that indeed $\lim_{p\to\infty} \kappa_p(\rho) = \lim_{p\to\infty} \chi_p(\rho)$. Now, going back to Theorem~\ref{hierarchies_chi_kappa_infty_thm}, and in particular to~\eqref{hierarchies_chi_kappa_infty}, we see that this implies that in fact 
\bb
\lim_{p\to\infty} E_{\kappa, p}(\rho) = \lim_{p\to\infty} E_{\chi, p}(\rho) = E_\kappa^\infty(\rho) = \ecost(\rho)\,, 
\label{convergence_SM}
\ee
i.e.\ the four central quantities appearing in~\eqref{hierarchies_chi_kappa_infty} coincide for all states. We will re-derive this equalities with some additional guarantees on the speed of convergence while proving Theorem~\ref{convergence_thm} below.

The identity in~\eqref{convergence_SM} is remarkable because it provides a precise connection between the zero-error PPT cost, given by the regularisation of $E_\kappa$ and thus by a limit $n\to\infty$ over the number of copies $n$, and the limits \emph{of the hierarchies}, intended as limits $p\to\infty$ on the hierarchy level. The fact that there would be any sort of connection between the asymptotic limit in the number of copies and that in the hierarchy level was a priori totally unclear, and it should be regarded as one of the neater results of our approach. However, Eq.~\eqref{convergence_SM} is still not completely satisfactory from a computational standpoint, because it does not provide any rigorous guarantee on the speed of convergence to the limit. Having such guarantees (either on the limit in $n$ or on that in $p$) would turn~\eqref{convergence_SM} into an \emph{algorithm}, because it would tell us at which $n$ or $p$ we would need to stop in order to achieve a certain approximation. If that could be done, then, due to the fact that $E_\kappa\big(\rho^{\otimes n}\big)$, $E_{\chi,p}(\rho)$, and $E_{\kappa,p}(\rho)$ are all computable via SDPs, we would know how to approximate $\ecost(\rho)$ up to an arbitrary accuracy. In short, whether the algorithm we have described works or not depends on whether or not we can say something about the speed of convergence to the limit in either $n$ or $p$. Furthermore, even if the algorithm does exist, then whether it is efficient or not depends on the actual speed of convergence to the limit, either in $n$ or in $p$. 

The key observation underpinning our entire approach is that while there seems to be no easy way of achieving either of the above two goals when dealing with the limit in the number of copies $n$, the situation changes dramatically for the better when considering the limit in the hierarchy level $p$. Not only will we be able to give universal, uniform bounds on the convergence speed, but the resulting convergence will also turn out to be exponentially fast (Theorem~\ref{convergence_thm}), resulting in an \emph{efficient} algorithm to calculate the zero-error PPT entanglement cost on all states (Theorem~\ref{efficient_algorithm_thm}).

\subsubsection{The local qubit case}

Proposition~\ref{convergence_primitive_prop} also implies an immediate simple solution to the problem of computing $\ecost(\rho_{AB})$ when either $A$ or $B$ is a single-qubit system. In this case, it turns out that both hierarchies collapse at the first level ($p=1$), entailing that $E_\chi = E_\kappa = \ecost$ already gives the true zero-error PPT cost.

\begin{manualcor}{\ref{qubit_case_cor}}
Let $AB$ be a bipartite quantum system in which either $A$ or $B$ is a single qubit, i.e.\ $d=\min\{|A|,|B|\}=2$. Then for all states $\rho=\rho_{AB}$ it holds that
\bb
\ecost(\rho) = E_{\chi}(\rho) = E_{\kappa}(\rho)\, , 
\label{qubit_case_SM}
\ee
where $E_\chi$ and $E_\kappa$ are defined by~\eqref{E_chi_SM} and~\eqref{E_kappa_SM}, respectively.
\end{manualcor}

\begin{note}
Remember that we identify $E_{\chi} = E_{\chi,1}$ (cf.~\eqref{E_chi_SM} and~\eqref{E_chi_p_SM}) and $E_\kappa = E_{\kappa,1}$ (cf.~\eqref{E_kappa_SM} and~\eqref{E_kappa_q_SM}).
\end{note}

\begin{proof}
Writing down~\eqref{convergence_primitive_SM} for $d=2$ and $p=1$ shows that $\kappa_1(\rho) \leq \chi_1(\rho)$. Combining this with Theorem~\ref{hierarchies_chi_kappa_infty_thm} then implies the claim.
\end{proof}

\begin{rem}
The above Corollary~\ref{qubit_case_cor} does not tell us what happens to the zeroth level of the $\chi$-hierarchy, namely, the logarithmic negativity. An older result by Ishizaka~\cite{Ishizaka2004} states that all \emph{two-qubit} states have zero bi-negativity, implying, by Lemma~\ref{normalisation_lemma}, that
\bb
|A| = |B| = 2\qquad \Longrightarrow\qquad \ecost(\rho) = E_N(\rho) = \log_2 \big\|\rho^\Gamma\big\|_1 \qquad \forall\ \rho = \rho_{AB}\, .
\ee
When e.g.\ $|A|=2$ but $|B|>2$, Corollary~\ref{qubit_case_cor} guarantees that $\ecost = E_\chi = E_\kappa$, but it would be even better if one could establish a closed-form expression for this quantity. One way to achieve this, for instance, would be by generalising Ishizaka's result so as to encompass all qubit-qudit systems. If that could be done, then again we would find that $\ecost = E_N$, which would be our sought closed-form expression. We leave a full understanding of the role of the logarithmic negativity in the qubit-qudit case as an open problem.
\end{rem}


\subsection{Exponential convergence: proof of Theorem~\ref{convergence_thm}} \label{subsec_exp_convergence}

This subsection is devoted to the proof of our first main result, Theorem~\ref{convergence_thm}, which will be seen to be a relatively straightforward consequence of the key Proposition~\ref{convergence_primitive_prop}.

\begin{manualthm}{\ref{convergence_thm}} \it
For all bipartite states $\rho=\rho_{AB}$ on a system of minimal local dimension $d\coloneqq \min\left\{|A|,|B|\right\}\geq 2$, and all positive integers $p\in \N^+$, it holds that
\bb
E_{\chi,p}(\rho) \leq E_{c,\,\ppt}^{\exact}(\rho) \leq E_{\kappa,p}(\rho) \leq E_{\chi,p}(\rho) + \log_2 \frac{1}{1-\left(1 - \frac{2}{d}\right)^p}\, .
\label{efficient_algorithm_key_inequality_SM}
\ee
In particular,
\bb
\ecost(\rho) = E_\kappa^\infty(\rho) = \lim_{p\to\infty} E_{\chi,p}(\rho) = \lim_{p\to\infty} E_{\kappa,p}(\rho)\, ,
\label{convergence_simplified_SM}
\ee
with the convergence in $p$ being exponentially fast uniformly on $\rho$.
\end{manualthm}

\begin{proof}
We give here a slightly modified and extended version of the argument presented in the main text. The reason for doing so is to incorporate also the convergence of the $\kappa$-hierarchy in a single statement. To deal with the $\chi$- and $\kappa$-hierarchies simultaneously, it is useful to set by convention $\kappa_0(\rho)\coloneqq d$. 
Now, define the numbers
\bb
\e_p(\rho) \coloneqq 1 - \frac{\chi_p(\rho)}{\kappa_p(\rho)}\, ,
\label{epsilon_p_SM}
\ee
which satisfy $\e_p\in [0,1]$ for all $p\in \N$ due to Proposition~\ref{relation_hierarchies_prop}. Note that this definition differs from that employed in the main text (Eq.~\eqref{convergence_thm_proof_eq2}) --- in particular, the quantity in~\eqref{epsilon_p_SM} is larger than that in~\eqref{convergence_thm_proof_eq2}. We can now write
\bb
1 &\leqt{(i)} \frac{d}{2}\,\frac{\chi_p(\rho)}{\kappa_p(\rho)} - \left(\frac{d}{2} -1\right) \frac{\chi_{p-1}(\rho)}{\kappa_p(\rho)} \\
&\leqt{(ii)} \frac{d}{2}\,\frac{\chi_p(\rho)}{\kappa_p(\rho)} - \left(\frac{d}{2} -1\right) \frac{\chi_{p-1}(\rho)}{\kappa_{p-1}(\rho)} \\
&\eqt{(iii)} 1 - \frac{d}{2}\,\e_p + \left(\frac{d}{2} -1\right)\e_{p-1}\, .
\label{epsilon_p_key_inequality}
\ee
The justification of the above steps is as follows: the inequality in~(i) is just a rephrasing of that established by Proposition~\ref{convergence_primitive_prop}, obtained by diving both sides by $\kappa_p(\rho)$ (remember that $\kappa_p(\rho)\geq 1$ by Remark~\ref{nonzero_rem}); in~(ii) we observed that $\kappa_p(\rho) \leq \kappa_{p-1}(\rho)$ due to Proposition~\ref{relation_hierarchies_prop} (this is also true for $p=1$ due to Lemma~\ref{normalisation_lemma}) and $d\geq 2$; finally, in~(iii) we employed the definition of $\e_p$ in~\eqref{epsilon_p_SM}.

We now see that~\eqref{epsilon_p_key_inequality} can be readily massaged into
\bb
\e_p \leq \left(1-\frac{2}{d}\right) \e_{p-1}\, .
\ee
Iterating this $p$ times and using the fact that $\e_0\leq 1$ and $\e_p\geq 0$ yields immediately
\bb
0\leq \e_p \leq \left(1-\frac{2}{d}\right)^p ,
\ee
which in turn can be rephrased as
\bb
\kappa_p(\rho) \leq \frac{\chi_p(\rho)}{1- \left(1-\frac{2}{d}\right)^p}\, .
\ee
Taking the logarithms of both sides proves the last inequality in~\eqref{efficient_algorithm_key_inequality_SM}. The first three inequalities are already known from Theorem~\ref{hierarchies_chi_kappa_infty}. 

Taking the limit $p\to\infty$ establishes also the convergence of both hierarchies to the true zero-error PPT entanglement cost (Eq.~\eqref{convergence_simplified_SM}). Note that the speed of convergence depends on $d$ but not on $\rho$.
\end{proof}

The above result is central in our approach, because it provides a quantitative guarantee on what level of the hierarchies we have to resort to in order to obtain a prescribed approximation of the true value of the zero-error PPT entanglement cost. Before we proceed with the description of an algorithm that builds on this observation, we take note of two remarkable consequences of Theorem~\ref{convergence_thm}, namely, the full addivity of the zero-error PPT entanglement cost and its continuity. 

Intuitively, the former means that the cheapest way to generate many copies of a state of the form $\rho_{AB}\otimes \omega_{A'B'}$, where $A$ and $A'$ belong to Alice, and $B$ and $B'$ to Bob, is to manufacture many copies of $\rho_{AB}$ and $\omega_{A'B'}$ \emph{separately}: in other words, there is no advantage to be gained in considering joint protocols. The continuity of the zero-error PPT cost, instead, allows us to estimate the cost of a state that is close enough to one for which the cost is known. The proofs of both of these facts are paradigmatic examples of how the knowledge we have gathered so far allows us to say a lot about the problem of zero-error PPT entanglement dilution even without a closed-form expression for the corresponding cost.

\begin{cor}
The zero-error entanglement cost under PPT operations is fully tensor additive, i.e.\ for all pairs of bipartite states $\rho_{AB}, \omega_{A'B'}$ it holds that
\bb
E_{c,\,\ppt}^\exact\left(\rho_{AB} \otimes \omega_{A'B'}\right) = E_{c,\,\ppt}^\exact\left(\rho_{AB}\right) + E_{c,\,\ppt}^\exact\left(\omega_{A'B'}\right) .
\ee
\end{cor}

\begin{proof}
It suffices to take the limit $p\to\infty$ of~\eqref{E_chi_additivity_SM} using~\eqref{convergence_simplified_SM}.
\end{proof}

\begin{cor}[(Continuity of the zero-error PPT entanglement cost)]
Let $\rho,\rho'$ be two quantum states on a bipartite system $AB$ with minimal local dimension $d\coloneqq \min\{|A|,|B|\}$. Let $\e\coloneqq \frac12 \left\|\rho-\rho'\right\|_1$. Then
\bb
\left|2^{\ecost(\rho)} - 2^{\ecost(\rho')} \right| \leq 2d\e\, ,
\label{continuity_exponential_cost}
\ee
and therefore
\bb
\left|\ecost(\rho) - \ecost(\rho') \right| \leq \log_2\left(1+2d\e\right) \leq 2\,(\log_2 e)\, d\e\, .
\label{continuity_cost}
\ee
\end{cor}

\begin{proof}
It suffices to take the limit $p\to\infty$ of~\eqref{continuity_F} and~\eqref{continuity_f} written for $F=\chi_p$ and $f=E_{\chi,p}$, using Theorem~\ref{convergence_thm}, and in particular~\eqref{convergence_simplified_SM}, to simplify the left-hand sides.
\end{proof}

\subsection{An explicit algorithm to compute the zero-error PPT cost: proof of Theorem~\ref{efficient_algorithm_thm}}

We will now exploit the above Theorem~\ref{convergence_thm} to prove the last of our main result, Theorem~\ref{efficient_algorithm_thm}. To this end, we will design an efficient algorithm that computes the zero-error PPT entanglement cost up to any desired accuracy.

\begin{manualthm}{\ref{efficient_algorithm_thm}}[(Efficient algorithm to compute the cost)] \it
Let $AB$ be a bipartite quantum system of total dimension $D \coloneqq |AB| = |A| |B|$ and minimal local dimension $d\coloneqq \min\{|A|,|B|\}$. Then there exists an algorithm that for an arbitrary state $\rho = \rho_{AB}$ computes $2^{\ecost(\rho)}$ up to a multiplicative error $\e$ --- and hence, a fortiori, $\ecost(\rho)$ up to an additive error $\e$ --- 
in time
\bb
\pazocal{O}\left((dD)^6\, \mathrm{polylog}\big(p,D,d,1/\e\big) \right) = \pazocal{O}\left((d D)^{6+o(1)} \operatorname{polylog}(1/\e) \right) .
\label{efficient_algorithm_time_SM}
\ee
\end{manualthm}


\begin{proof}
To start, choose
\bb
p \coloneqq \ceil{\frac{\ln (2d/\e)}{\ln \frac{d}{d-2}}}\, ,
\label{efficient_algorithm_proof_eq2}
\ee
so that
\bb
\left(1-\frac2d\right)^p \leq \frac{\e}{2d}\, .
\ee
Due to~\eqref{efficient_algorithm_key_inequality_SM}, we then know that
\bb
1-\frac{\e}{2d} \leq \frac{\chi_p(\rho)}{2^{\ecost(\rho)}} \leq 1\, ,
\ee
which implies that
\bb
\left| \chi_p(\rho) - 2^{\ecost(\rho)} \right| \leq \frac{\e}{2d}\, 2^{\ecost(\rho)} \leq \frac{\e}{2}\, ,
\label{efficient_algorithm_proof_eq3}
\ee
where the last equality is a consequence of the fact that $\ecost(\rho)\leq \log_2 d$ for all $\rho$ by Lemma~\ref{normalisation_lemma}.

Now, $\chi_p(\rho)$ is given by an SDP and hence it can be computed efficiently up to additive error $\e/2$. We will look at the time it takes to carry out this computation in a moment, but for the time being assume that it yields an estimator $\widehat{\chi}_p(\rho)$ with the property that 
\bb
\left| \widehat{\chi}_p(\rho) - \chi_p(\rho) \right| \leq \frac{\e}{2}\, .
\label{efficient_algorithm_proof_eq4}
\ee
Then by combining~\eqref{efficient_algorithm_proof_eq3} and~\eqref{efficient_algorithm_proof_eq4} one obtains that
\bb
\left| \widehat{\chi}_p(\rho) - 2^{\ecost(\rho)} \right| \leq \e\, .
\ee
This yields the claimed computation of $\ecost(\rho)$ up to an additive error $\e$. To estimate the error one incurs when computing $\ecost(\rho)$, it suffices to observe that
\bb
\left|\,\ecost(\rho) - \log_2 \widehat{\chi}_p(\rho)\, \right| \leq (\log_2 e) \left| \,\widehat{\chi}_p(\rho) - 2^{\ecost(\rho)}\, \right| \leq \e \log_2 e\, ,
\ee
so the error on $\ecost(\rho)$ is also $\e$, up to a constant. In the above derivation, we used the estimate $|a-b|\leq (\log_2 e) \big|2^a - 2^b\big|$, valid for all $a,b\geq 0$.

We still need to estimate the running time of the SDP that computes the estimator $\widehat{\chi}_p(\rho)$. To do that, we use the best known SDP solvers, which have time complexity~\cite{LSW2015} (see also the quantum-friendly review in~\cite{vanApeldoorn2020})
\bb
\pazocal{O}\left(m \left(m^2 + n^\omega + mns\right) \mathrm{polylog}(m,n,R,1/\e) \right)
\label{best_time_complexity_SDPs}
\ee
once the SDP is put in the standard form
\bb
\begin{array}{rl}
\text{max.} & \Tr CX \\
\text{s.t.} & X\geq 0, \\
& \Tr A_j X \leq b_j\, ,\quad j=1,\ldots, m\, ,
\end{array}
\label{SDP_standard_form}
\ee
where $\|A_j\|_\infty,\|C\|_\infty\leq 1$ for $j=1,\ldots,m$, all matrices are $n\times n$, $R$ is any upper bound on the optimal value, $\omega\in [2,2.373)$ is the matrix multiplication exponent, and $s$ is the sparsity, i.e.\ maximum number of non-zero entries in any row of the input matrices $A_j,C$.

What we have to do now is to cast the SDP for $\chi_p(\rho)$ (Eq.~\eqref{chi_p_1_SM}) in the standard form~\eqref{SDP_standard_form}. This can be done by 
restricting the variable $X$ to be of the form
\bb
X \longrightarrow \bigoplus_{i=0}^p \begin{pmatrix} S_i & S_{i-1}^\Gamma \\ S_{i-1}^\Gamma & S_i \end{pmatrix} ,
\label{forcing_X}
\ee
where $S_{-1} = \rho$. To see why, note that for a matrix of the above block form positive semi-definiteness implies that
\bb
\begin{pmatrix} S_i & S_{i-1}^\Gamma \\ S_{i-1}^\Gamma & S_i \end{pmatrix} \geq 0\qquad \forall\ i=0,\ldots, p\, ,
\ee
which by a simple unitary rotation can be rephrased as
\bb
\frac12 \begin{pmatrix} \id & \id \\ \id & -\id \end{pmatrix} \begin{pmatrix} S_i & S_{i-1}^\Gamma \\ S_{i-1}^\Gamma & S_i \end{pmatrix} \begin{pmatrix} \id & \id \\ \id & -\id \end{pmatrix}^\dag = \begin{pmatrix} S_i + S_{i-1}^\Gamma & 0 \\ 0 & S_i - S_{i-1}^\Gamma \end{pmatrix} \geq 0\qquad \forall\ i=0,\ldots, p\, ,
\ee
i.e.\ $S_i \pm S_{i-1}^\Gamma \geq 0$ for all $i=0,\ldots,p$, matching the positive semi-definite constraints in~\eqref{chi_p_1_SM}.

The above reasoning tells us that it is a good idea to try to enforce the structure~\eqref{forcing_X} on $X$. This fixes the size of the matrices to be 
\bb
n=2(p+1)D = \pazocal{O}(pD)\, ,
\label{estimate_n}
\ee
where $D=|A||B|$ is the total dimension. The problem therefore becomes that of ensuring that $X$ has the structure in~\eqref{forcing_X}. To this end, it is useful to think of $X$ as a block matrix of size $(p+1)\times (p+1)$, where each block has size $2D\times 2D$. We need to make sure that:
\begin{enumerate}[(1)]
\item all the $p(p+1)$ off-diagonal blocks are zero;
\item inside each diagonal block $X_{ii}$ ($i=0,\ldots,p$), which can in turn be thought of as a $2\times 2$ block matrix $X_{ii} = \lsmatrix H_i & K_i \\ L_i & M_i \rsmatrix$ with each block having size $D$, we have that $H_i=M_i$ and $K_i=L_i$;
\item for all $i=1,\ldots,p$, $K_i = H_{i-1}^\Gamma$; and
\item $K_0 = \rho^\Gamma$.
\end{enumerate}
All of these constraints are equalities instead of inequalities, but any equality can be written as two inequalities, and we do not count factors of $2$. How many linear equality constraints on $X$ do we need?
\begin{enumerate}[(1):]
\item $p(p+1)D^2$ constraints;
\item $2(p+1)D^2$ constraints;
\item $p D^2$ constraints;
\item $D^2$ constraints.
\end{enumerate}
We are therefore dealing with a total of $(p+1)(p+3)D^2$ linear equality constraints. It is therefore clear that we can set 
\bb
m=\pazocal{O}(p^2D^2)\, .
\label{estimate_m}
\ee

Now, constraints~(1)--(3) are of the form $\Tr A X = 0$, where $A$ has only two non-zero elements on the same row or column, while the constraints~(4) are of the form $\Tr AX=b$, where $A$ has a single non-zero entry (equal to $1$). Note also that since we want $\Tr C X = \Tr S_p$, $C$ needs to be block diagonal, with the first $p$ diagonal blocks equal to zero and the last one equal to $\id/2$ (all blocks are of size $2D$). We can thus set
\bb
s = \pazocal{O}(1)\, .
\label{estimate_s}
\ee
As for $R$, by Lemma~\ref{normalisation_lemma} we already know that 
\bb
R=d
\label{estimate_R}
\ee
is an upper bound on $\chi_p(\rho)$. Plugging the values~\eqref{estimate_n}--\eqref{estimate_R} into~\eqref{best_time_complexity_SDPs} yields a time complexity
\bb
\pazocal{O}\left((pD)^6\, \mathrm{polylog}\big(p,D,d,1/\e\big) \right) &= \pazocal{O}\left((dD)^6\, \mathrm{polylog}\big(p,D,d,1/\e\big) \right) \\
&= \pazocal{O}\left((d D)^{6+o(1)} \operatorname{polylog}(1/\e) \right),
\label{time_complexity_E_chi_p}
\ee
where we used the fact that $p = \pazocal{O}\left(d \log (d/\e)\right)$ due to~\eqref{efficient_algorithm_proof_eq2}.
\end{proof}

\begin{rem}
It can be verified that the time complexity of solving the SDP that defines the negativity $E_N$ is $\pazocal{O}\big(D^{6+o(1)} \mathrm{polylog}(1/\e)\big)$. While this is smaller than the complexity for computing $\ecost$, it is only marginally so --- the difference being a mere factor $\mathrm{poly}(d)$.
\end{rem}


\subsection{Open problem: hierarchy collapse}

Although we showed that evaluating the limit $\lim_{p\to\infty}E_{\chi,p}(\rho)$ to any desired precision is computationally not much more demanding than evaluating $E_{\chi,p}$ for a fixed level of the hierarchy $p$, one may still wonder --- perhaps from an analytical or aesthetic standpoint --- whether computing the limit is truly necessary, or whether the hierarchy collapses at a certain point and there exists a finite $p^\star$ such that $E_{\chi,p} = E_{\chi,p^\star} \; \forall p \geq p^\star$. An analogous question can be asked for the $\kappa$-hierarchy.

One immediate consequence of Proposition~\ref{convergence_primitive_prop} is that, if $E_{\chi,p} (\rho) = E_{\chi,p-1} (\rho)$ for some $p$, then $E_{\kappa,p} (\rho) \leq E_{\chi, p}(\rho)$, which would imply a complete collapse of both hierarchies, that is,
\begin{equation}\begin{aligned}
	\ecost(\rho) = E_{\chi,p} (\rho) = E_{\kappa,p} (\rho)\, ,
\end{aligned}\end{equation}
due to Proposition~\ref{relation_hierarchies_prop}.

We were not able to confirm nor disprove that this happens. However, numerical evidence suggests that the hierarchies collapse, and indeed they do so already at the 
second level of the $\chi$-hierarchy. 
This leads us to posit the following conjecture.

\begin{cj}
For all states $\rho = \rho_{AB}$, the $\chi$-quantities defined by~\eqref{chi_p_1_SM} satisfy 
that $E_{\chi,3} (\rho) = E_{\chi,2}(\rho)$. As a consequence, the $\chi$- and $\kappa$-hierarchies collapse and $\ecost(\rho) = E_{\chi,2}(\rho) = E_{\kappa,3}(\rho)$.
\end{cj}

In fact, we were not even able to find a gap between $E_{\chi,2}$ and $E_{\kappa,2}$, so one could even make the stronger conjecture that $\ecost(\rho) = E_{\chi,2}(\rho) = E_{\kappa,2}(\rho)$ holds for all states, which would entail that both hierarchies collapse at the second level. For the sake of obtaining a single-letter, limit-free formula for $\ecost$, 
it would be very interesting to resolve this question in future work. However, as remarked in the main text, doing so would only amount to a $\operatorname{poly}(d)$ reduction of the computational complexity of calculating the zero-error PPT entanglement cost $\ecost$.

\section{From zero error to very small error}

Zero-error entanglement manipulation tasks may seem on the surface un-physical, because nothing in nature happens with zero error, and small, undetectable errors should therefore always be included into the picture. Here we will do precisely that, and we will find that \emph{if those errors are assumed to be sufficiently small (but non-zero), the resulting entanglement dilution rates are the same as in the zero-error case.} In this context, `sufficiently small' could mean for example `going to $0$ super-exponentially in the number of copies of the given state'.

This is no different from what happens in classical information theory, where the capacity for communicating on a noisy classical channel with error probability going to $0$ super-exponentially equals its zero-error capacity, defined by Shannon in his landmark 1956 paper~\cite{Shannon-zero-error}. The argument to prove this claim seems to be part of the folklore in information theory, and it was brought to our attention by Andreas Winter. It goes as follows. A classical channel $\NN$ with input alphabet $\XX$ and output alphabet $\pazocal{Y}$ can be thought of as a transition matrix $N(y|x)\geq 0$, with $\sum_y N(y|x)=1$ for all $x\in \XX$. A code to communicate on $\NN^{\otimes n}$ is a list of words $x^n_1,\ldots, x^n_M$, where $M=2^{\ceil{nR}}$, and $R$ is the rate of the code. Suppose that for a fixed $n$ a word in $\{1,\ldots, M\}$ is drawn at random. Whatever the decoder is, the probability of making an error can be estimated from below as follows. Pick two distinct $i,j\in\{1,\ldots, M\}$. For a fixed $y^n\in \pazocal{Y}^n$, call $p_i \coloneqq \NN^{\otimes n}(y^n|x_i^n)$ and $p_j \coloneqq \NN^{\otimes n}(y^n|x_j^n)$ the probabilities that these two words get transformed into the same output word $y^n\in \pazocal{Y}^n$. When that happens, even if the decoding party knows that the message was either $i$ or $j$ all they can do is to guess $i$ or $j$ with a maximum likelihood rule, i.e.\ $i$ with probability $\frac{p_i}{p_i+p_j}$ and $j$ with probability $\frac{p_j}{p_i+p_j}$. The total probability of error is thus at least
\bb
\frac1M\, p_i\, \frac{p_j}{p_i+p_j} + \frac1M\, p_j\, \frac{p_i}{p_i+p_j} = \frac1M\, \frac{2p_ip_j}{p_i+p_j} \geq \frac1M \min\{p_i, p_j\}\, .
\ee
Now, if the right-hand side is not zero for some $y^n$, then it must be at least $t^n/M$, where $t\coloneqq \min_{x,y:\, N(y|x)>0} N(y|x)$. Therefore the probability of error of the whole process (encoding, transmission, and decoding) is at least $t^n 2^{-\ceil{nR}}$, which cannot go to $0$ super-exponentially. This completes the summary of the state of affairs in classical communication over noisy channels.

To arrive at analogous conclusions in the case of PPT entanglement dilution, we first need to fix some terminology. Our first task is to design a way to control the errors incurred in an arbitrary entanglement manipulation protocol. This can be done by introducing a rate-error pair achievability region for entanglement manipulation under a given set of free operations $\FF$, defined as follows.

\begin{Def}
For two bipartite states $\rho=\rho_{AB}$ and $\omega=\omega_{AB}$ and a class of free operations $\FF$, we say that the rate-error pair $(r,s)$ is achievable for the transformation $\rho \to \omega$ with operations in $\FF$ if there exists a sequence of protocols $(\Lambda_n)_{n\in \N_+}$, with $\Lambda_n\in \FF\left(A^nB^n\to {A'}^n{B'}^n\right)$ for all $n\in \N_+$, such that
\bb
\liminf_{n\to\infty} \left\{-\frac1n \log_2 \frac12 \left\|\Lambda_n\big(\rho^{\otimes \ceil{rn}}\big) - \omega^{\otimes n}\right\|_1\right\} \geq s\, .
\ee
\end{Def}

In our case, we care about the case where $\FF=\ppt$ is the set of PPT operations, the initial state is the ebit $\Phi_2$, and the final state is an arbitrary $\rho$. In this case, a rate-error pair that is achievable for $\Phi_2\to \rho$ under PPT operations is also called an achievable rate-error pair for PPT entanglement dilution to $\rho$. The main result of this section is as follows.

\begin{thm} \label{small_error_thm}
Let $\rho=\rho_{AB}$ be an arbitrary finite-dimensional bipartite state. If $(r,s)$ is achievable for PPT entanglement dilution to $\rho$, then
\bb
\max\big\{r,\, \log_2 d-s\big\} \geq \ecost(\rho)\, .
\ee
In particular, for all $s > \log_2 d - \ecost(\rho)$ the rate-error pair $(r,s)$ is achievable if and only if $r\geq \ecost(\rho)$.
\end{thm}

Before proving the above result, let us discuss its implications. What this shows is that if the error is required to decay \emph{sufficiently fast}, i.e.\ faster than $2^{-n \big(\log_2 d - \ecost(\rho)\big)}$ asymptotically, then the optimal rate of entanglement dilution coincides with its zero-error value. In other words, once a certain threshold in $s$ is passed, it does not make a difference for the rate $r$ whether $s$ increases further, even if it goes all the way to infinity. Note that Theorem~\ref{small_error_thm} applies to two suggestive special cases, namely where (i)~we require the error to decay faster than $d^{-n}$; and therefore also when (ii)~we require a super-exponential decay law. In both of these cases, the relevant optimal rate of PPT entanglement dilution coincides with the zero-error PPT entanglement cost $\ecost(\rho)$.

\begin{proof}[Proof of Theorem~\ref{small_error_thm}]
%
%
Consider a sequence of PPT operations $(\Lambda_n)_{n\in \N_+}$ such that
\bb
\e_n \coloneqq \frac12 \left\|\Lambda_n\big(\Phi_2^{\otimes \floor{rn}}\big) - \rho^{\otimes n} \right\|_1
\ee
exhibits the asymptotic decay rate
\bb
\liminf_{n\to\infty} \left\{-\frac1n \log_2 \e_n\right\} \geq s\, .
\ee
Then,
\bb
2^{\floor{rn}} &= \kappa_1\left( \Phi_2^{\otimes \floor{rn}} \right) \\
&\geqt{(i)} \kappa_1\left( \Lambda_n\left(\Phi_2^{\otimes \floor{rn}} \right)\right) \\
&\geqt{(ii)} \kappa_1\big(\rho^{\otimes n}\big) - 2d^n \e_n \\
\label{small_error_proof_chain}
\ee
where (i)~is a consequence of the monotonicity of $\kappa_1 = 2^{E_\kappa}$ under PPT operations (Proposition~\ref{strong_monotonicity_prop}), while (ii)~is an application of Proposition~\ref{continuity_properties_prop}, and in particular of~\eqref{continuity_F} with $F=\kappa_1$. Therefore,
\bb
1 + \max\big\{\floor{rn},\, 1 + n \log_2 d + \log_2 \e_n\big\} &= \log_2 \left(2 \max\left\{2^{\floor{rn}},\, 2d^n\e_n\right\} \right) \\
&\geq \log_2\left(2^{\floor{rn}} + 2d^n \e_n \right) \\
&\geq \log_2 \kappa_1\big(\rho^{\otimes n}\big) \\
&= E_\kappa\big(\rho^{\otimes n}\big)\, .
\ee
Dividing by $n$ and taking the $\limsup$ of both sides as $n\to\infty$ yields precisely
\bb
\max\{r,\, \log_2 d -s\} &\geq \limsup_{n\to\infty} \frac1n \left( 1 + \max\big\{\floor{rn},\, 1 + n \log_2 d + \log_2 \e_n\big\} \right) \\
&\geq \limsup_{n\to\infty} \frac1n\, E_\kappa\big(\rho^{\otimes n}\big) \\
&= \lim_{n\to\infty} \frac1n\, E_\kappa\big(\rho^{\otimes n}\big) \\
&= E_\kappa^\infty(\rho)\, ,
\ee
concluding the proof.
\end{proof}

\end{document}